\documentclass[sigconf,screen,nonacm]{acmart} 
\AtBeginDocument{%
  }

\usepackage[unicode]{hyperref}

\hypersetup{
	pdfstartview=FitH,
	bookmarksnumbered=true,
	bookmarksopen=true,
	colorlinks,
	linkcolor={violet!90!black},
	citecolor={blue!60!black},
	urlcolor={blue!80!black}
}

\usepackage{amsmath,amsfonts,amsthm}

\usepackage{amssymb}

\usepackage{dsfont}
\usepackage{framed}
\usepackage{soul}
\usepackage{setspace}

\usepackage[linesnumbered,ruled,vlined]{algorithm2e}
\usepackage{enumitem}
\usepackage{url}

\SetKwInput{KwInput}{Input}                
\SetKwInput{KwOutput}{Output}              
\SetKw{Break}{break}
\SetKw{And}{\textbf{and}}

\newcommand{\eps}{\varepsilon}

\newcommand{\Renyi}{{R{\'e}nyi}}
\newcommand{\privdpo}{\texttt{PrivDPO}}

\usepackage{mathtools}

\theoremstyle{plain}
\newtheorem{thm}{Theorem}[section]

\newtheorem{lemma}{Lemma}[section]
\newtheorem{defn}{Definition}[section]

\newtheorem{prop}{Proposition}[section]

\theoremstyle{definition}
\newtheorem{exmp}{Example}[section]

\theoremstyle{remark}

\usepackage{bm}
\usepackage{graphicx}
\usepackage{textcomp}
\usepackage{xcolor}

\usepackage{subfigure}
\makeatletter
\newcommand\captionof[1]{\def\@captype{#1}\caption}
\makeatother

\usepackage{balance}
\usepackage{multirow}

\usepackage{booktabs}

\usepackage[most]{tcolorbox}

\let\originalmiddle=\middle
\def\middle#1{\mathrel{}\originalmiddle#1\mathrel{}}

\begin{document}

\title{Private Direct Preference Optimization for LLM Alignment}

\author{Yangfan Jiang}
\authornote{Work done during an internship at Tongyi Lab, Alibaba Group.}
\affiliation{%
 \institution{National University of Singapore}
 \country{Singapore}
}
\email{jyangfan@u.nus.edu}

\author{Fei Wei}
\affiliation{%
 \institution{Alibaba Group}
 \city{Hangzhou}
 \country{China}
}
\email{feiwei@alibaba-inc.com}

\author{Ergute Bao}
\affiliation{%
 \institution{Inria}
 \city{Saclay}
 \country{France}
}
\email{ergute.bao@inria.fr}

\author{Xiaokui Xiao}
\affiliation{%
 \institution{National University of Singapore}
 \country{Singapore}
}
\email{xkxiao@nus.edu.sg}

\author{Yaliang Li}
\affiliation{%
 \institution{Alibaba Group}
 \city{Bellevue}
 \country{United States}
}
\email{yaliang.li@alibaba-inc.com}

\author{Bolin Ding}
\affiliation{%
 \institution{Alibaba Group}
 \city{Bellevue}
 \country{United States}
}
\email{bolin.ding@alibaba-inc.com}

\begin{abstract}
Direct preference optimization (DPO) is now a standard method for aligning large language models (LLMs) using human preference data. Each DPO example contains a prompt and a pair of candidate model responses. While prompts and responses are often public or model-generated, the relative preference between responses reflects subjective judgments and can reveal sensitive attributes of annotators or end users. Off-the-shelf privacy-preserving approaches are not well matched to this structure, leading to unnecessary noise injection and biased updates in training.

In this paper, we formalize \emph{preference privacy}, a label-DP-style privacy notion for DPO that protects only the relative preference between candidate responses, assuming an adversary who already knows the prompt and responses. We then design \privdpo{}, a DPO variant that enforces preference privacy while remaining compatible with large-scale LLM training. Our main observation is that, for neighboring examples differing only in their preference signal, the gradient difference lies on a one-dimensional preference axis determined solely by the text; all preference information flows through this axis. \privdpo{} adds calibrated randomness only along this axis via an unbiased randomized rescaling of the DPO objective, avoiding per-example gradient operations. Our experiments on three alignment benchmarks and three LLM families show that \privdpo{} consistently achieves strong privacy-utility trade-offs compared with privacy-preserving baselines.
\end{abstract}

\maketitle

\section{Introduction}\label{sec:intro}
Large language models (LLMs) have become foundational building blocks for many real-world applications, including chatbots, coding agents, and health assistants~\cite{bubeck2023sparks,touvron2023llama,gpt4,comanici2025gemini}. 
At the same time, LLMs can produce outputs that are harmful, biased, or otherwise misaligned with human values~\cite{liang2021towards,bai2022training,bai2022constitutional,ouyang2022training}, which limits their usage in security- and safety‑critical settings. 

To address this, modern systems align a pre-trained model using human preference data, most prominently via reinforcement learning from human feedback (RLHF)~\cite{bai2022constitutional,stiennon2020learning,ouyang2022training}, ensuring that model behaviors better reflect human values and produce outputs that are less harmful, biased, or unsafe. 
A representative technique in this line of RLHF work is \emph{direct preference optimization} (DPO)~\cite{rafailov2023direct}, which has been widely adopted for aligning production-scale LLMs~\cite{qwen2024qwen25technicalreport,dubey2024llama,jiang2024mixtral,abdin2024phi,hui2024qwen2}, and is also supported by OpenAI's fine-tuning API for aligning GPT-4.1-series models~\cite{openai-api-platform}. 

The main idea of DPO is to fine-tune a pre-trained LLM on datasets containing human preference signals, thereby calibrating the model's behavior based on positive and negative human feedback. 
Specifically, each DPO training example is a triplet $(x,y_w,y_l)$, where $x$ denotes the input prompt, $y_w$ the preferred response aligned with human values, and $y_l$ the less preferred response that may be harmful or unhelpful to the input prompt $x$.
In practice, such datasets are typically collected by prompting an LLM to produce a pair of candidate responses, after which human annotators or end-users select the preferred option.

\subsection{The Case for Preference Privacy}\label{subsec:intro-motivateion}
While DPO is effective for LLM alignment, its reliance on human preference data creates privacy risks, especially in vertical domains such as healthcare, finance, and legal assistance.
The preference signal, i.e., whether an annotator prefers $y_w$ over $y_l$, can reveal personal judgments or sensitive attributes.
For example, preferences on politically charged or ethically sensitive prompts may expose an individual's stance, beliefs, or group membership, creating risks such as harassment or employment discrimination if mishandled. 
Beyond individual risk, mishandling of sensitive preference data during the DPO process can undermine public trust and damage the credibility of the organizations developing and deploying LLMs. 
Protecting preference information is therefore important for both individual privacy and responsible LLM deployment.

A natural approach is to apply \emph{differential privacy} (DP)~\cite{dinur2003revealing,dwork2004privacy,dwork2006calibrating,dwork2006our}, which protects individual records by adding calibrated noise during training~\cite{abadi2016deep}. However, standard DP is overly conservative in this context, because it protects the entire triplet $(x,y_w,y_l)$, including the input prompt $x$, both candidate responses $y_w$ and $y_l$, and the associated preference signal. In practice, however, input prompts and candidate responses in DPO datasets are typically non-private, since they are generated by LLMs or drawn from publicly available sources. For example, in the TL;DR summarization dataset~\cite{stiennon2020learning} used by OpenAI, the public Reddit post serves as the prompt, the pre-trained LLM generates candidate summaries, and annotators indicate which summary they prefer. In such cases, the truly sensitive element is in the individuals' preference signals, i.e., which response they prefer, rather than the text content of either the prompt or the LLM-generated responses. As a result, applying standard DP would inject noise for non-sensitive components, degrading utility without matching the actual privacy need.

Although prior work on label DP~\cite{esmaeili2021antipodes,ghazi2021deep,busa2023label,esfandiari2022label,jiang2024protecting} addresses a similar challenge in supervised learning, such methods rely on additional assumptions about label distributions or task structures and are therefore not direct drop-in replacements for standard DPO at LLM scale. Other potential approaches, such as randomized response techniques~\cite{warner1965randomized,erlingsson2014rappor}, can privatize preference labels through random flipping, but they introduce bias into DPO optimization and degrade alignment performance. There remains a lack of practical privacy-preserving techniques specifically designed for the unique gradient landscape of DPO.

\subsection{Contributions}\label{subsec:contributions}
Motivated by this, we formalize \emph{preference privacy}, which protects user preference information embedded in human feedback data. 
At a high level, it guarantees that, regardless of the underlying preference signal, outcomes derived from the data remain indistinguishable. Building on this privacy notion, we propose \privdpo{}, an effective and scalable DPO method for LLM alignment with formal preference privacy guarantees.

This notion follows the same indistinguishability structure as label DP, instantiated for preference annotations; our main contribution is to realize this guarantee efficiently for DPO-based LLM alignment. While preference privacy can be viewed as a natural relaxation of DP, achieving it in practical LLM alignment, especially at scale, presents two significant challenges: 
(i) introducing as little randomness as possible to satisfy the privacy notion without unnecessary utility loss, which requires accurately identifying and carefully perturbing the sensitive component in the parameter space that encodes human preferences; and
(ii) ensuring scalability for practical LLM optimization, which requires non-trivial algorithmic design to effectively balance privacy, utility, and efficiency.

\privdpo{} addresses these challenges by leveraging several theoretical insights of preference privacy in the DPO setting, derived from an in-depth analysis of the DPO objective under preference privacy constraints. 
Specifically, we examine the privacy-leakage surface in the DPO training process and identify how user preference information is encoded in model gradients. 
Our analysis reveals that the gradient difference between two otherwise identical samples, differing only in their preference annotations, aligns with a specific axis that captures the preference signal between $y_w$ and $y_l$ given the prompt $x$. 
This axis is determined solely by the textual content of $x$, $y_w$, and $y_l$, without relying on any sensitive preference information. 
In other words, user preference information flows through a one-dimensional subspace of the LLM parameter space, which we refer to as the \emph{preference axis}.
This observation implies that it is sufficient to introduce randomness only along this preference axis to protect user preference information, rather than perturbing the entire model gradient space. 

However, directly instantiating this idea in LLM training, e.g., by computing the preference axis vector for each gradient and perturbing it to mask preference information, is impractical. 
For example, in a production-level LLM with over 30B parameters, the gradient of a single DPO example already requires more than 100~GB of GPU memory, not including the additional cost of optimizer states, e.g., in the case of Adam~\cite{DBLP:journals/corr/KingmaB14}. Moreover, retrieving per-example gradients is incompatible with existing large-scale training frameworks, which shard model gradients and optimizer states across GPUs for scalability and throughput, and therefore cannot obtain a complete single-example gradient on demand. 
As a result, direct gradient manipulation is infeasible for production-scale LLMs.

To overcome these challenges, we further analyze the DPO objective under preference privacy constraints while accounting for model utility. 
We find that unbiased perturbation of gradients along the preference axis is mathematically equivalent to privatizing the DPO objective via \emph{an asymmetric randomized rescaling mechanism.}
This insight leads to our final solution, \privdpo{}, a simple yet effective alignment method that enforces preference privacy without explicit gradient manipulation. 
The core idea is to perturb the DPO objective for each training example using a tailored random rescaling strategy, ensuring preference privacy while preserving unbiasedness. 
This design is fully compatible with modern large-scale LLM training frameworks and requires only minimal code changes. 
Importantly, our analysis is model-agnostic and depends only on the form of the DPO objective, making \privdpo{} broadly applicable to all types of LLM architectures. 
We further show that \privdpo{} yields unbiased gradients with favorable error bounds under the preference privacy model, avoiding the bias of RR-based input perturbation and the high-dimensional noise of standard DP mechanisms such as DP-SGD.

Beyond its algorithmic contribution, our work also offers several practical insights. Empirical results show that \privdpo{} consistently outperforms existing privacy-preserving baselines across diverse LLM families, including  Llama, Pythia, and Qwen, and across model sizes ranging from 3B to 32B parameters. Notably, we demonstrate, for the first time to our knowledge, that it is practical to align large-scale LLMs with up to 32B parameters under a rigorous privacy guarantee for user preferences without notable compromises of training efficiency or model quality.

\section{Preliminaries}
\label{sec:background}

\subsection{Direct Preference Optimization}
\label{sec:background-dpo}
We begin by formalizing large language models (LLMs), then describe LLM alignment via RLHF with a reward model, and finally introduce DPO as a reward-model-free alternative. 
\subsubsection{LLMs}
We consider an autoregressive language model $\pi_{\theta}$ parameterized by $\theta \in \mathbb{R}^d$, where $d$ scales into the billions for modern LLMs. 
Given a prompt $x$, the model generates a response sequence $y = [y_1, y_2, \dots, y_N]$ token by token, with probability
\begin{align*}
    \pi_{\theta}(y\mid x) = \prod_{i=1}^{N} \pi_{\theta}(y_i \mid x,y_{<i}),
\end{align*}
where $y_{<i}:=y_{1:i-1}$ %
denotes previously generated tokens.

\subsubsection{Preference Optimization}\label{subsec:dpo-notations}
We assume a preference dataset $\mathcal{D}=\{(x, y_w, y_l)\}$, where each example is a triplet consisting of a prompt $x$ and a pair of responses $(y_w, y_l)$. 
We follow the standard convention in the DPO literature~\cite{rafailov2023direct}, where the ordering encodes the human preference $y_w \succ y_l$. That is, $(y_w, y_l)$ is not merely a set of candidates, but an ordered pair in which the first element denotes the preferred response. This representation is widely adopted and avoids introducing an additional preference variable.

In our setting, the sensitive information lies precisely in this ordering, i.e., the relative preference between the two responses, while the textual content of $x, y_w, y_l$ is assumed to be non-sensitive.

Let $\pi_{\rm ref}$ denote a {reference model}, which is typically a pre-trained or supervised fine-tuned LLM that has not yet been aligned with human preferences. 
The goal of RLHF-based alignment is to fine-tune a model $\pi_{\theta}$ initialized from $\pi_{\mathrm{ref}}$ to maximize the expected reward while remaining close to the reference model. Formally, this can be written as a KL-regularized optimization problem: 
\begin{align}
    \max_{\pi_\theta} \mathbb{E}_{x\sim\mathcal{D}, y\sim\pi_{\theta}(\cdot \mid x)}\big[ r (x,y)\big] - \beta D_{\rm KL}\left( \pi_{\theta} \middle\Vert \pi_{\rm ref} \right), \label{eq:rlhf-obj}
\end{align}
where $r(x,y)$ is a reward learned from pairwise comparisons based on Bradley-Terry model~\cite{bradley1952rank}. The parameter $\beta$ controls the divergence penalty from $\pi_{\rm ref}$, preventing the aligned model from drifting too far from the base LLM and thereby mitigating the risk of model collapse (e.g., catastrophic degradation in model behavior). 
We refer interested readers to~\cite{schulman2015trust,christiano2017deep,schulman2017proximal,rafailov2023direct,stiennon2020learning,ziegler2019fine,meng2024simpo,ethayarajh2024model,azar2024general} for further details. 

Under mild assumptions, Rafailov et al.~\cite{rafailov2023direct} show that the reward can be implicitly represented in terms of $\pi_{\theta}$:
\begin{align}
    r(x,y)=\beta \log \frac{\pi_{\theta}(y\mid x)}{\pi_{\rm ref}(y\mid x)} + \beta \log Z(x),\label{eq:dpo-reward}
\end{align}
where $Z(x)$ is a partition function independent of $y$. 
This formulation allows the pairwise preference between two responses to be modeled as a probabilistic comparison, following the Bradley-Terry assumption~\cite{bradley1952rank} that the probability of preferring $y_w$ over $y_l$ is given by a logistic function of the difference between their implicit rewards. 
From the implicit reward in Eq.~\eqref{eq:dpo-reward}, the DPO objective for a preference triplet $t = (x, y_w, y_l)$ can be derived as~\cite{rafailov2023direct}:
\begin{align}
    &\mathcal{L}_{\rm DPO}(t;\pi_{\theta};\pi_{\rm ref}) =  \nonumber\\
    &\quad  -\log \sigma \left( \beta \log\frac{\pi_{\theta}(y_w\mid x)}{\pi_{\rm ref}(y_w\mid x)} - \beta \log\frac{\pi_{\theta}(y_l\mid x)}{\pi_{\rm ref}(y_l\mid x)} \right),\label{eq:dpo-los}
\end{align}
where $\sigma(\cdot)$ denotes the sigmoid function. Intuitively, the DPO objective pushes $\pi_{\theta}$ toward assigning higher probability to $y_w$ over $y_l$, while penalizing deviations from $\pi_\mathrm{ref}$.

In practice, DPO training starts from a base model, often with a warm-up supervised fine-tuning (SFT) stage on $(x,y)$ pairs from $\mathcal{D}$. The resulting $\pi_{\rm ref}$ better matches the preference data distribution, making DPO optimization more stable and effective~\cite{rafailov2023direct,ethayarajh2024model,meng2024simpo}.

\subsection{Differential Privacy}
\label{sec:background-DP}
Differential privacy (DP)~\cite{dinur2003revealing,dwork2004privacy,dwork2006calibrating,dwork2006our} is a rigorous framework for protecting individuals' sensitive information. It has been extensively studied, and more recently, in the context of deep learning and LLM fine-tuning~\cite{abadi2016deep,tramer2021differentially,de2022unlocking,bao2025unlocking,lilarge,yudifferentially}. We review its definition. 
\begin{defn}[Differential Privacy~\cite{dwork2006calibrating,dwork2006our}]\label{def:dp}
A randomized mechanism $\mathcal{M}$ satisfies $(\eps, \delta)$-DP, if for any two neighboring datasets $D, D'$ differing in only one record, and for any subset of possible outputs $\mathcal{O}\subseteq Range(\mathcal{M})$, we have
\begin{equation}
\Pr\left[\mathcal{M}\left(D\right)\in\mathcal{O}\right]\leq e^\eps \Pr\left[\mathcal{M}\left(D'\right)\in\mathcal{O}\right]+\delta.\label{eq:def-dp}
\end{equation}
\end{defn}
The parameter $\eps$ is known as the \emph{privacy budget} that controls the privacy and utility trade-off. 
Smaller values of $\eps$ and $\delta$ imply that the output distributions of $\mathcal{M}$ on any pair of neighboring datasets are harder to distinguish, thereby providing stronger privacy guarantees. 
Given a non-private algorithm, the basic idea to achieve DP is to introduce calibrated noise into the algorithm procedure~\cite{dinur2003revealing}, ensuring that altering a single record does not significantly affect the outcome distribution. 
In general, smaller privacy parameters require injecting larger DP noises into the algorithm procedure, which can in turn degrade the accuracy of the algorithm.

\vspace{1mm}
\noindent\textbf{Deep learning with DP.} 
We briefly review standard approaches for achieving DP in deep learning, most notably DP-SGD~\cite{abadi2016deep} and its variants. DP aims to ensure that any single training record has only a negligible influence on the output distribution of the algorithm. 
To this end, DP-SGD bounds each record's influence by clipping the $\ell_2$ norm of per-example gradients and adding noise to the clipped gradients, ensuring that gradients computed from neighboring records are statistically indistinguishable.

DP-SGD has been widely applied in deep learning~\cite{abadi2016deep,tramer2021differentially,de2022unlocking}, and more recently in LLM fine‑tuning~\cite{bao2025unlocking,lilarge,yudifferentially}. 
However, as noted earlier and elaborated in Section~\ref{subsec:theoretical-analysis}, its privacy model and algorithmic design are not well aligned with DPO-based LLM alignment. This motivates our tailored notion of preference privacy (Section~\ref{sec:preference-privacy}) and the corresponding mechanism design (Section~\ref{sec:mechanism-design}). 

\section{Formalizing Preference Privacy}\label{sec:preference-privacy}
This section presents the definition and privacy properties of $\epsilon$-preference privacy, followed by a review of existing first-cut solutions that satisfy this privacy notion. The privacy definition adopts the same notion of indistinguishability used in DP and can therefore be viewed as a natural relaxation of DP, with the protected information specifically tailored to the DPO setting. 

In what follows, we first define neighboring DPO training records that differ only in their preference annotations, and then introduce the privacy notion based on this definition.

\begin{defn}[Preference Neighboring]\label{def:neighboring}
Given two DPO training triplets $t$ and $t'$, we say that $t$ and $t'$ are preference neighbors if and only if they share the same input prompt $x$ and output response pair $\{y_w, y_l\}$, but differ in their preference annotations. 
\end{defn}

\begin{exmp}\label{exmp:neighbor}
Let $t=(x, y_w, y_l)$ be a DPO training example with the preference signal $y_w \succ y_l$.
Define $t'=(x, y_l, y_w)$ with $y_w \prec y_l$, which is identical to $t$ except for the reversed preference signal.
Then $t$ and $t'$ are preference neighbors.
\end{exmp}

Let $\mathcal{T}$ denote the domain of the DPO training example $t$. We define preference privacy as follows.

\begin{defn}[$\epsilon$-Preference Privacy]\label{def:preference-priv}
A randomized mechanism $\mathcal{M}:\mathcal{T}\mapsto\mathcal{R}$ satisfies $\epsilon$-preference privacy if, for all pairs of preference neighboring DPO examples $t$ and $t'$, and for any subset of outputs $O\subseteq \mathcal{R}$, it holds that
\begin{equation}\label{eq:defn-distp}
    \Pr[\mathcal{M}(t)\in O]\leq e^\epsilon\cdot \Pr[\mathcal{M}(t')\in O].
\end{equation}
\end{defn}

We note that Definition~\ref{def:preference-priv} adopts a pure $\epsilon$-style guarantee, in contrast to the $(\epsilon,\delta)$-DP definition in Definition~\ref{def:dp}. This distinction is intentional. In our setting, preference privacy is defined over a binary preference signal, which admits mechanisms (including ours and RR-based baselines) that satisfy pure $\epsilon$-style guarantees without requiring a $\delta$ relaxation.  
By contrast, standard approaches such as DP-SGD operate in high-dimensional parameter spaces and typically rely on Gaussian noise, which only provides $(\epsilon,\delta)$-DP guarantees. We therefore present $(\epsilon,\delta)$-DP in Definition~\ref{def:dp} for completeness, while focusing on pure $\epsilon$-preference privacy for our mechanism design and analysis.

Next, we present the key properties of $\epsilon$-preference privacy and provide its semantic interpretation, clarifying the scope of protected preference information and its privacy implications in DPO.

\subsection{Privacy Properties and Semantics}\label{subsec:privayc-semantics}

\noindent\textbf{Privacy properties and sequential composition.} 
Since the definition of $\epsilon$-preference privacy is inspired by and closely follows well-established privacy notions such as differential privacy and Pufferfish privacy~\cite{kifer2014pufferfish}, it inherits several important privacy properties shared by those notions. 
In particular, $\epsilon$-preference privacy satisfies \emph{convexity} and \emph{transformation invariance} (referred to as the post-processing property in the context of DP), which are two fundamental privacy properties deemed essential for a privacy notion to provide a sound and intuitive guarantee~\cite{kifer2010towards,kifer2012axiomatic,kifer2014pufferfish}.

\begin{prop}[Convexity]
    Let $\mathcal{M}_1$ and $\mathcal{M}_2$ be any two randomized mechanisms with independent sources of randomness and both satisfy $\epsilon$-preference privacy. Denote by $\widetilde{\mathcal{M}}$ a randomized protocol that runs $\mathcal{M}_1$ with probability $q$ and runs $\mathcal{M}_2$ with probability $1-q$. Then, for any $q\in(0,1)$, $\widetilde{\mathcal{M}}$ also satisfies $\epsilon$-preference privacy.
\end{prop}

The convexity property ensures that introducing an additional source of uncertainty, i.e., by randomly selecting between mechanisms with the same preference privacy guarantee, never weakens that guarantee.

\begin{prop}[Transformation Invariance]\label{corollary:transformation-inv}
    Let $\mathcal{M}:\mathcal{T}\mapsto\mathcal{R}$ be a mechanism that satisfies $\epsilon$-preference privacy. Then, for any algorithm $\mathcal{A}:\mathcal{R}\mapsto \mathcal{H}$ that cannot access the random bits of $\mathcal{M}$, 
it holds that 
$\mathcal{A}\circ \mathcal{M}:\mathcal{T}\mapsto \mathcal{H}$ also satisfies $\epsilon$-preference privacy.
\end{prop}

The implication of transformation invariance is that, as long as $\mathcal{M}$ satisfies $\epsilon$-preference privacy and its random bits are not leaked, the outputs $\mathcal{A}\circ \mathcal{M}(t)$ and $\mathcal{A}\circ \mathcal{M}(t')$ are indistinguishable for any pair of preference-neighboring $t$ and $t'$, from the view of any computationally unbounded adversary that is simulated by algorithm $\mathcal{A}$. Consequently, the adversary cannot gain more information about the underlying preference than the privacy definition allows, regardless of how the adversary processes the observed outcome from $\mathcal{M}$.

\begin{prop}[Sequential Composition]\label{prop:seq-composition}
Let $\mathcal{M}_1:\mathcal{T}\mapsto\mathcal{R}_1$ satisfy $\epsilon_1$-preference privacy. Let $\mathcal{M}_2:\mathcal{T}\times\mathcal{R}_1\mapsto\mathcal{R}_2$ be such that, for every fixed $r_1\in\mathcal{R}_1$, the mechanism $t\mapsto \mathcal{M}_2(t,r_1)$ satisfies $\epsilon_2$-preference privacy. 
Then 
$
\mathcal{M}_{1,2}(t):=(\mathcal{M}_1(t),\,\mathcal{M}_2(t,\mathcal{M}_1(t)))
$
satisfies $(\epsilon_1+\epsilon_2)$-preference privacy.
\end{prop}

Here, the adaptive mechanism $\mathcal{M}_2$ is allowed to depend on the previously released output $\mathcal{M}_1(t)$; the condition above requires that, after fixing this auxiliary input, $\mathcal{M}_2$ still satisfies $\epsilon_2$-preference privacy as a mechanism over $t$.

Sequential composition ensures that the theoretical privacy leakage bound for mechanisms with a preference privacy guarantee increases linearly as the mechanism is repeatedly run on the same underlying data. This is particularly useful for accounting the total privacy cost when training involves multiple epochs.

\vspace{1mm}
\noindent\textbf{Interpretation of privacy semantics.}
The preference privacy notion implicitly assumes an adversary with complete knowledge of the text content in a DPO training example, i.e., the text of $x$, $y_w$ and $y_l$, who seeks to infer the preference between $y_w$ and $y_l$. The privacy implication of such privacy notion is that it theoretically bounds the adversary's \emph{information gain} about the underlying human preference after observing the mechanism's output, formally stated in the following theorem.

\begin{prop}[Privacy Semantics]\label{lemma:bound-attacker}
Let $b\in\{0,1\}$ denote the private preference bit, where $b=1$ indicates the input is $t_1=(x,y_w,y_l)$ and $b=0$ indicates input is $t_0=(x,y_l,y_w)$. 
Let $\mathcal{M}$ be an $\epsilon$-preference private mechanism producing output $R=\mathcal{M}(t_b)$.

Let $\mathcal{A}$ be any adversary that analyzes the mechanism's output $R$, and define $S=\mathcal{A}(R)$ as the adversary's observation. 
Then for any observation $s$, it holds that
\begin{align*}
    e^{-\epsilon}\leq \frac{\Pr\left[ S=s\mid b=1 \right]}{\Pr\left[ S=s\mid b=0 \right]}\leq e^\epsilon,
\end{align*}
and therefore for any adversary's prior belief denoted by $p_{\rm adv}:=\Pr\left[b=1\right]$, we have
\begin{align*}
    e^{-\epsilon}\frac{p_{\rm adv}}{1-p_{\rm adv}} \leq \frac{\Pr\left[ b=1 \mid S=s \right]}{\Pr\left[ b=0\mid S=s \right]}\leq e^\epsilon\frac{p_{\rm adv}}{1-p_{\rm adv}}.
\end{align*}
\end{prop}

\begin{proof}[Proof sketch]
    The proof follows the privacy semantics of Pufferfish privacy~\cite{kifer2014pufferfish}. 
    Let the adversary's randomization be the conditional density $\mu(s\mid r)=\Pr\left[ S=s \mid R=r\right]$, describing how the mechanism's output $r$ is mapped to the observation $s$. 
    Since $\mathcal{M}$ satisfies $\epsilon$-preference privacy, we have 
    \[ e^{-\epsilon}\leq\Pr\left[R=r\mid b=1\right] / \Pr\left[R=r\mid b=0\right]\leq e^\epsilon \] 
    for all $r$. By the law of total probability, averaging these bounds over $\mu(s\mid r)$ preserves the inequality, thus 
    \[e^{-\epsilon}\leq \Pr\left[S=s\mid b=1\right] / \Pr\left[ S=s\mid b=0 \right] \leq e^\epsilon.\] Applying Bayes' rule then gives the stated bound. 
\end{proof}

$\epsilon$-preference privacy also implies a dataset-level variant.

\begin{defn}[$\epsilon$-Preference Privacy, Dataset-Level]\label{def:centralized-pref-priv}
A randomized mechanism $\widehat{\mathcal{M}}$ that takes a dataset as input satisfies dataset-level $\epsilon$-preference privacy if, for any two DPO datasets $D, D'$ differing in only one preference annotation (i.e., by flipping the preference of a single triplet $t \in D$), and for any subset of outputs $\mathcal{O}$, it holds that
\[
\Pr[\widehat{\mathcal{M}}(D) \in \mathcal{O}] \leq e^{\epsilon} \Pr[\widehat{\mathcal{M}}(D') \in \mathcal{O}].
\]
Here, $\widehat{\mathcal{M}}(D)$ denotes the (randomized) output of the mechanism when run on dataset $D$, which in our setting corresponds to the transcript of privatized gradients produced during training.
\end{defn}

Specifically, consider a mechanism $\widehat{\mathcal{M}}$ that adaptively processes the dataset sequentially and outputs a transcript of privatized gradients $\{\mathcal{M}(t_i)\}$ for $t_i \in D$. Here, \emph{adaptive} means that each invocation $\mathcal{M}(t_i)$ may depend on previously released outputs $\{\mathcal{M}(t_j)\}_{j<i}$. 
For neighboring datasets $D,D'$ differing in one preference annotation, let $k$ be the affected index. Releases before step $k$ have identical inputs; at step $k$, the inputs are preference neighbors and the release is $\epsilon$-preference private; after step $k$, all remaining examples are unchanged and later releases depend on the flipped preference only through the protected transcript. Therefore, $\widehat{\mathcal{M}}(D)$ satisfies dataset-level $\epsilon$-preference privacy, assuming each example is processed once in the sequential pass.

In our algorithm (Section~\ref{subsec:algorithm-design}), each DPO gradient release satisfies $\epsilon$-preference privacy. Therefore, the entire private DPO procedure, which updates model parameters using only such privatized outputs, also satisfies dataset-level $\epsilon$-preference privacy.

Finally, note that preference privacy can be viewed as a relaxation of DP. 
A natural first-cut approach is therefore to apply existing DP mechanisms, such as DP-SGD or randomized preference flipping. 
However, these methods are not tailored to DPO under preference privacy and suffer from fundamental technical limitations. 
In particular, they either perturb the entire parameter space unnecessarily or introduce significant bias, leading to degraded alignment performance. 
We defer a detailed discussion of these first-cut solutions to Appendix~\ref{appendix:first-cut-solution}, and demonstrate their limitations through theoretical and empirical analyses in Sections \ref{subsec:theoretical-analysis} and \ref{sec:experiments}, respectively.

\subsection{Threat Model and Applicability}\label{subsec:threat-model} 
We next clarify the threat model, scope of protection, and applicability of preference privacy, and discuss the settings in which enforcing preference privacy is appropriate.

\vspace{1mm}
\noindent\textbf{Threat model.}
We consider an adversary who observes trained models and intermediate checkpoints, similar to the canonical privacy-preserving learning framework of DP-SGD~\cite{abadi2016deep}. In addition, the adversary knows the {prompt and the candidate responses.} The adversary's goal is to {infer an individual user's preference between two candidate responses.}

Note that preference privacy is not designed as an ad-hoc defense against any particular classes of attacks on LLM training data, such as membership inference~\cite{shokri2017membership} or data reconstruction attacks~\cite{carlini2021extracting}. 
Instead, it provides an attacker-agnostic, worst-case guarantee via bounding how much additional information about any individual's preference can be revealed through training. Roughly speaking, $\epsilon$-preference privacy guarantees that the released information can change the adversary's \emph{posterior belief} about a preference by at most a multiplicative factor of $e^{\epsilon}$, limiting the incremental leakage attributable to the use of private preference data.

\vspace{1mm}
\noindent\textbf{Applicability.} 
Preference privacy is tailored to DPO-based alignment workflows in which the human preference signal is sensitive, while prompts and responses are non-sensitive and can be reasonably assumed to be known to the adversary. 
As discussed earlier in Section~\ref{subsec:intro-motivateion}, this assumption holds in many practical DPO settings. 
In common DPO datasets, prompts are typically drawn from public, non-proprietary sources, and candidate responses are generated by a pre-trained LLM. 
For example, in summarization~\cite{stiennon2020learning} and dialogue alignment benchmarks~\cite{bai2022training}, only the annotator's preference may reveal sensitive personal judgments or attributes. 
In such settings, preference privacy captures the dominant privacy risk and provides a more accurate and less conservative guarantee than standard DP.

\vspace{1mm}
\noindent\textbf{Limitations.}
Preference privacy does not apply to all alignment scenarios. If prompts or responses themselves contain sensitive or proprietary information, then protecting only the preference signal is insufficient. That said, our approach is complementary to existing DP-based methods~\cite{abadi2016deep,yudifferentially,lilarge,bao2025unlocking} and is intended for settings in which sensitivity is primarily confined to the preference signal.

\vspace{1mm}
\noindent\textbf{Relation to label DP.}
Preference privacy is closely related to label DP~\cite{esmaeili2021antipodes,ghazi2021deep}, where features are public and only labels are sensitive. In DPO, the prompt and candidate responses play the role of public features, while the preference annotation is the private label. Thus, our preference-neighboring relation can be viewed as a label-DP-style formulation specialized to DPO preference data. 
This differs from triplet-level DP, where neighboring datasets may differ in the entire DPO triplet, including the prompt, responses, and preference annotation. 
Our contribution is not the privacy definition itself, but the mechanism: PrivDPO exploits the DPO objective to perturb a one-dimensional preference-dependent scalar rather than the raw label or the full gradient.

\section{Private Direct Preference Optimization}
\label{sec:mechanism-design}

Off-the-shelf first-cut solutions for protecting preference annotations fall into two categories. The first applies standard DP-SGD to the DPO objective, which protects the entire triplet $(x,y_w,y_l)$ and therefore gives a strong, triplet-level guarantee but requires high-dimensional gradient perturbation. The second uses randomized response (RR)-based input perturbation to randomize the preference signal before DPO training, which matches the preference-level privacy unit but introduces bias into the DPO optimization. 
We briefly discuss these baseline solutions in Section~\ref{subsec:privayc-semantics}, with details deferred to Appendix~\ref{appendix:first-cut-solution}.
Both issues severely degrade alignment performance even under moderate privacy guarantees. The proposed \privdpo{} is designed to overcome these limitations by perturbing the intermediate preference-dependent term in the DPO objective, rather than the full gradient or the raw preference label, while offering formal preference privacy. 

In what follows, we first outline the key intuition behind \privdpo{}, followed by its design details and a formal analysis of its privacy and utility guarantees.

\subsection{Rationale}\label{subsec:rationale}
The design of \privdpo{} is motivated by a fundamental observation about how preference information is captured by DPO training gradients.
Consider a pair of neighboring DPO training examples, $t=(x, y_w, y_l)$ and $t'=(x, y_l, y_w)$, which differ only in their preference signal. 
For a fixed LLM parameter $\pi_{\theta}$, we define the DPO gradient difference as
\begin{align*}
     \Lambda = \frac{\partial \mathcal{L}_{\rm DPO} (t;\pi_{\theta};\pi_{\rm ref})}{\partial \theta} - \frac{\partial \mathcal{L}_{\rm DPO} (t';\pi_{\theta};\pi_{\rm ref})}{\partial \theta} \in \mathbb{R}^d.
\end{align*}
This DPO gradient difference $\Lambda$ has a unique geometric structure that can be exploited for privacy-preserving mechanism design. 
Specifically, $\Lambda$ lies along a specific axis in parameter space determined by the textual content of the triplet, i.e., $x$, $y_w$, and $y_l$; we term this the \emph{preference axis}.

This observation reveals a unique pattern of privacy leakage in DPO training. 
At a high level, privacy requires the algorithm's output distribution to remain indistinguishable when the sensitive component, i.e., the preference annotation, is flipped. 
For DPO, such a flip does not arbitrarily change the gradient in the full parameter space; for a fixed prompt and response pair, it only changes the scalar coefficient along the preference axis, while the orthogonal components remain unchanged. 

Intuitively, this means that the adversary's ability to infer preference information is limited to this one-dimensional subspace; the remaining orthogonal subspace is free from privacy concerns. 
This insight enables a mechanism that injects randomness only along the sensitive preference axis, rather than perturbing the input (i.e., raw preference labels) or the intermediate state (i.e., full gradients).

\begin{figure}[!t]
    \centering
    \hspace{-3.5mm} 
    \includegraphics[width=0.9\linewidth, keepaspectratio]{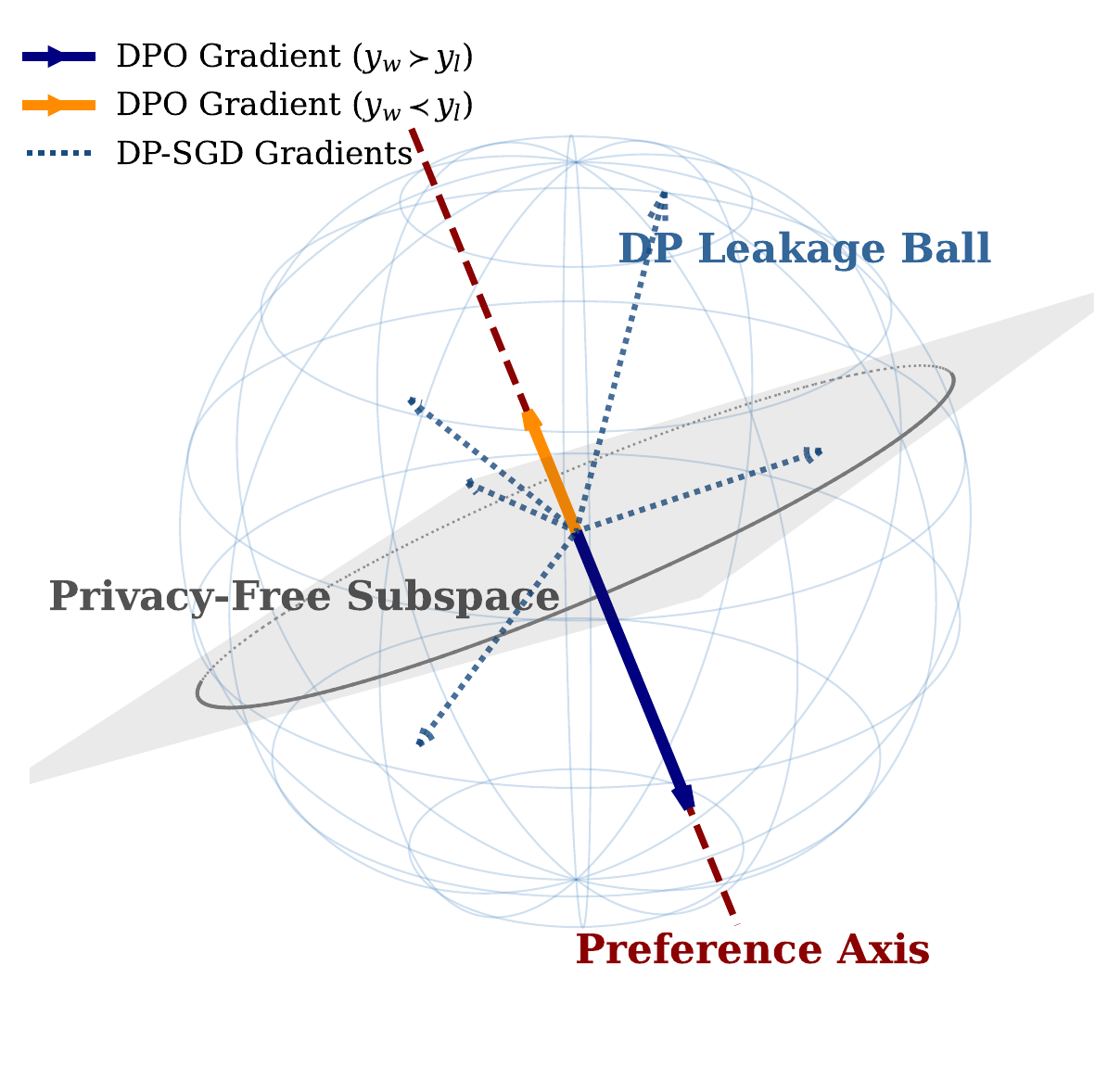}
    \vspace{-9mm} 
    \caption{Illustration of the privacy leakage surface.}
    \label{fig:intuition}
    \vspace{-3mm}
\end{figure}

\subsection{Understanding the Privacy Leakage Surface}\label{sec:sensitivity-analysis}
We now examine the geometric structure of the preference leakage surface in DPO.
Specifically, we analyze how an adversary might exploit the gradient to infer preference information from a DPO training example.
Let $\boldsymbol{v}:=\nabla_\theta\log\pi_\theta(y_w\mid x) - \nabla_\theta\log\pi_\theta(y_l\mid x)$. For a DPO example $t=(x,y_w,y_l)$, differentiating the DPO objective $\mathcal{L}_{\rm DPO}$ with respect to the LLM parameter $\theta$ yields
\begin{align}
    -\beta\cdot \sigma\left( \beta\log \frac{\pi_{\theta}(y_l \mid x)}{\pi_\text{ref}(y_l \mid x)} - \beta\log \frac{\pi_{\theta}(y_w \mid x)}{\pi_\text{ref}(y_w \mid x)}\right) \cdot \boldsymbol{v},\nonumber
\end{align}
where $\sigma(\cdot)$ denotes the sigmoid function. 
For $t'=(x,y_l,y_w)$, the corresponding gradient can be expressed as:
\begin{align*}
    -\beta\cdot \sigma\left( \beta\log \frac{\pi_{\theta}(y_w \mid x)}{\pi_\text{ref}(y_w \mid x)} - \beta\log \frac{\pi_{\theta}(y_l \mid x)}{\pi_\text{ref}(y_l \mid x)}\right) \cdot(- \boldsymbol{v}).
\end{align*}
Since the sigmoid function satisfies $\sigma(z)=1-\sigma(-z)$, defining $\psi:=\sigma\left( \beta\log \frac{\pi_{\theta}(y_l \mid x)}{\pi_\text{ref}(y_l \mid x)} - \beta\log \frac{\pi_{\theta}(y_w \mid x)}{\pi_\text{ref}(y_w \mid x)}\right)$, we have
\begin{align*}
    \Lambda 
    =& -\beta  \psi \cdot \Big[\nabla_\theta\log\pi_\theta(y_w\mid x) - \nabla_\theta\log\pi_\theta(y_l\mid x)\Big] \\
    & \;\; + \beta  (1-\psi) \cdot \Big[\nabla_\theta\log\pi_\theta(y_l\mid x) - \nabla_\theta\log\pi_\theta(y_w\mid x) \Big]  
    \\
    =& -\beta \cdot \Big[\nabla_\theta\log\pi_\theta(y_w\mid x) - \nabla_\theta\log\pi_\theta(y_l\mid x)\Big]=-\beta\boldsymbol{v}.
\end{align*}

The above expression shows that a preference flip changes the DPO gradient only along the preference axis $\boldsymbol{v}$. This is related to label-DP-style relaxations in its privacy semantics: the text content plays the role of public features, while the preference ordering is the sensitive annotation. The key difference is algorithmic. General label-DP mechanisms do not rely on such a gradient structure, whereas in DPO, for a fixed $(x,y_w,y_l)$, the two preference-neighboring gradients remain on the same one-dimensional line determined by $\boldsymbol{v}$. 
This property follows from the algebraic form of the DPO loss, rather than from model architecture, and therefore holds even when $\pi_\theta$ is a nonlinear LLM.

The geometric intuition is illustrated in Figure~\ref{fig:intuition}. Compared with standard DP, which perturbs full gradients, and label-DP-style input perturbation, which randomizes the raw preference label, \privdpo{} exploits the DPO-specific structure above: a preference flip only affects the gradient within the one-dimensional subspace spanned by $\boldsymbol{v}$. This motivates perturbing the intermediate preference-dependent scalar rather than the full gradient or the raw preference label.

\subsection{Algorithm Design}\label{subsec:algorithm-design}

\begin{algorithm*}[!t]
\caption{\small Private DPO via Randomized Objective Rescaling}
\label{alg:privdpo}
\KwInput{Dataset $D=\{(x^{(i)}, y_w^{(i)},y_l^{(i)})\}_{i=1}^{n}$; policy model $\pi_\theta$; reference model $\pi_{\rm ref}$; DPO hyperparameter $\beta$; DPO objective function $\mathcal{L}_{\rm DPO}$; privacy parameter $\epsilon$; learning rate $\eta$.}
\KwOutput{Model parameters ${\theta}$.}
$\mathcal{B}:=\{B_1,...,B_k\}$\tcp*{Randomly permute $D$ and partition it into disjoint mini-batches $\mathcal{B} = \{B_1, \ldots, B_k\}$}
\ForEach{$B \in \mathcal{B}$}{
    \ForEach{{\rm triplet} $t=(x,y_w,y_l)\in B$}{
    $\psi \gets \texttt{sg}\left[\sigma \left(\beta\log\frac{\pi_{\theta}(y_l\mid x)}{\pi_{\rm ref}(y_l \mid x)} - \beta\log\frac{\pi_{\theta}(y_w\mid x)}{\pi_{\rm ref}(y_w \mid x)}\right)\right]$ \tcp*{$\texttt{sg}[\cdot]$: stop-gradient operator; $\sigma(\cdot)$: sigmoid function} \label{algline:compute-psi}
    Sample $u$ uniformly at random from $[0,1]$\;
    $\tilde{w} \gets 
    \left\{
    \begin{array}{ll}
       1+\frac{1}{\psi}\left(1-\frac{e^\epsilon-2}{e^\epsilon-1}\right), & \text{if } u\leq \frac{e^\epsilon}{e^\epsilon+1},\\
      1-\frac{1}{\psi}\left(1+\frac{1}{e^\epsilon-1}\right), & \text{otherwise.}
    \end{array}
    \right.$\;\label{algline:random-weight}
    
    $\tilde{\mathcal{L}}_t \gets {\tilde{w}}\cdot \mathcal{L}_{\rm DPO}(t;\pi_\theta;\pi_{\rm ref})$ \tcp*{Rescale DPO objective with the random weight $\tilde{w}$} \label{algline:obj-perturb}
    $\boldsymbol{g}_t \gets \nabla_\theta \tilde{\mathcal{L}_t}$\; \label{algline:get-gradient}
    }
    $\boldsymbol{g}\gets\frac{1}{|B|}\sum_{t\in B} \boldsymbol{g}_t$\; \label{algline:agg-grad}
    $\theta \gets \theta - \eta \boldsymbol{g}$ \tcp*{Update model parameters by gradient descent}
}
\Return ${\theta}$\;
\end{algorithm*}

Consider an adversary who has access to the textual content of a DPO training example $(x, y_w, y_l)$, and is able to observe its corresponding gradient.
The adversary's goal is to infer whether $y_w \succ y_l$ or $y_w \prec y_l$.
The only exploitable clue for this inference is the direction of the gradient along the preference axis, which is uniquely defined by the neighboring preference pair
$\boldsymbol{v}:=\nabla_\theta\log\pi_\theta(y_w\mid x) - \nabla_\theta\log\pi_\theta(y_l\mid x)$. 
Intuitively, to prevent the adversary from distinguishing whether the input is $(x, y_w, y_l)$ or $(x, y_l, y_w)$, it suffices to randomize the directional information along the preference axis $\boldsymbol{v}$. 
This targeted perturbation conceals the preference information embedded in the gradient while avoiding the excessive utility loss caused by isotropic noise across the full parameter space.

However, the magnitude of the gradient along the preference axis encodes valuable alignment information that must be preserved as much as possible for training utility. 
Specifically, for a given preference $y_w \succ y_l$, the DPO gradient can be expressed as $-\beta \psi \boldsymbol{v}$, where $\beta$ is a DPO hyperparameter, and the scalar 
\begin{align*}
    \psi = \sigma\left( \beta\log \frac{\pi_{\theta}(y_l \mid x)}{\pi_\text{ref}(y_l \mid x)} - \beta\log \frac{\pi_{\theta}(y_w \mid x)}{\pi_\text{ref}(y_w \mid x)}\right)
\end{align*}
represents how strongly the model currently prefers $y_l$ over $y_w$. 
We refer to $\psi$ as the \emph{preference intensity}. 
When the model already assigns relatively higher likelihood to $y_w$ than to $y_l$, i.e., $\log \pi_\theta(y_w \mid x) > \log \pi_\theta(y_l \mid x)$, $\psi$ becomes relatively small, implying only a minor update step for that example. 
Hence, a desirable privacy-preserving DPO should retain this intensity information as faithfully as possible while ensuring preference privacy.

\vspace{1mm}
\noindent\textbf{Unbiased intensity perturbation mechanism.} 
To achieve this balance, we devise an unbiased perturbation mechanism on the preference intensity $\psi$ that satisfies $\epsilon$-preference privacy while maintaining the expected value of $\psi$. 
Unlike RR-based label-DP input perturbation, which randomizes the discrete preference annotation before computing the loss, our mechanism keeps the DPO objective intact up to the preference-dependent scalar $\psi$ and applies randomness directly to this intermediate quantity. Although the private annotation is binary, its effect on the DPO gradient appears as one of two scalar coefficients separated by a constant. 
Specifically, by exploiting the property of the sigmoid function $\sigma(z) = 1 - \sigma(-z)$, we can verify that the effect of flipping the preference on the DPO gradient can be viewed as shifting the preference intensity from $\psi$ to $\psi-1$. This corresponds to a change in the gradient from $-\beta \psi \boldsymbol{v}$ to $-\beta (\psi-1) \boldsymbol{v}$. Hence, given the underlying preference axis, the only exploitable information available to an adversary is the magnitude of the preference intensity, which can take exactly two possible values separated by a difference of $1$. 
This unique structure unlocks greater flexibility for designing private yet unbiased mechanisms than is available in the general numeric setting of local differential privacy \cite{duchi2018minimax,wang2019collecting}, where inputs range over a continuous interval (see Appendix~\ref{appendix:variant-privdpo}).

A direct approach would be to ensure the privacy guarantee by bounding the likelihood ratio of outputs within $e^\epsilon$, i.e., making the true value $\psi$ at most $e^\epsilon$ times more likely to be output than any possible alternative value $\psi'$ in the output range. 
However, such a direct approach introduces bias into the perturbed estimates. 
To address this issue, we design a mechanism that never outputs $\psi$ exactly, but instead samples around it asymmetrically to maintain unbiasedness. 
Specifically, with higher probability, it outputs a value slightly larger than $\psi$, and with lower probability, a value slightly less than $\psi - 1$.
By carefully calibrating these offsets, we ensure that the expectation of the perturbed output equals $\psi$, thereby preserving unbiasedness while satisfying $\epsilon$-preference privacy.
Formally, the mechanism is defined as:
\begin{align*}
    \mathcal{M}_{\rm unbias}(\psi)=    \left\{
    \begin{array}{ll}
       \psi + \frac{1}{e^\epsilon-1}, & \text{with probability } \frac{e^\epsilon}{e^\epsilon+1},\\
      \psi- 1-\frac{1}{e^\epsilon-1}, & \text{with probability } \frac{1}{e^\epsilon+1},
    \end{array}
    \right.
\end{align*}
where the offset is set to $\frac{1}{e^\epsilon-1}$ to ensure unbiasedness. 

Let $\mathcal{M}_{\rm grad}(t)= -\beta\mathcal{M}_{\rm unbias}(\psi)\boldsymbol{v}$ denote the randomized gradient generated by this process. 
It can be verified that (i) $\mathbb{E}[\mathcal{M}_{\rm unbias}(\psi)] = \psi$, and thus $\mathcal{M}_{\rm grad}$ is an unbiased estimator of the DPO gradient; and (ii) $\mathcal{M}_{\rm grad}$ satisfies the indistinguishability requirement in Definition~\ref{def:preference-priv}, thereby ensuring $\epsilon$-preference privacy. 
We defer the formal privacy and utility analysis to Section~\ref{subsec:theoretical-analysis}.

A remaining practical concern is that directly applying this perturbation requires fetching and manipulating each gradient individually.
For large-scale LLMs with tens of billions of parameters, this is both computationally and memory intensive: a single-precision gradient for a 32B-parameter model can exceed 100~GB of GPU memory, with additional overhead from model parameters and optimizer states, far beyond current GPU capacities. 
Moreover, modern training frameworks shard models, gradients, and optimizer states across GPUs and therefore do not natively support per-example gradient operations. 
Implementing per-example operations would require reverting to pure data parallelism, placing all gradients and states onto a single GPU, making training prohibitively slow and impractical for large models. 
To address these bottlenecks, we next introduce an equivalent but far more efficient approach that exploits the geometric structure of DPO gradients, avoiding explicit per-example gradient manipulation.

\vspace{1mm}
\noindent\textbf{Rescaling the DPO objective.} 
We observe that flipping the preference signal in DPO training data only changes the direction and scale of the resulting gradient. 
In other words, the pair of corresponding gradients are symmetric about the origin, when ignoring their preference intensities.
Perturbing the preference intensity $\psi$ is therefore equivalent to applying a randomized rescaling to the DPO objective itself before backpropagation. 
Formally, when perturbing the preference intensity from $\psi$ to $\tilde{\psi}$, the gradient changes from $-\beta \psi \boldsymbol{v}$ to $-\beta \tilde{\psi}\boldsymbol{v}$, which can be written as:
\begin{align*}
-\beta\tilde{\psi}\boldsymbol{v}=-\frac{\tilde{\psi}}{\psi}\psi \beta\boldsymbol{v}=\frac{\tilde{\psi}}{\psi}\nabla_{\theta} \mathcal{L}_{\rm DPO}=\nabla_{\theta} \left(\texttt{sg}\left[{\tilde{\psi}}{\psi^{-1}}\right]\mathcal{L}_{\rm DPO}\right),
\end{align*}
where $\texttt{sg}[\cdot]$ denotes the stop-gradient operator. 
Hence, perturbing the preference intensity of the gradient along the preference axis is equivalent to applying a randomized rescaling to the DPO objective. 

\vspace{1mm}
\noindent\textbf{The \privdpo{} algorithm.}
This insight motivates our final private DPO solution, \privdpo{}, presented in Algorithm~\ref{alg:privdpo}. 
We follow the standard DPO training protocol~\cite{rafailov2023direct}: the dataset is randomly permuted and processed sequentially in mini-batches for a single epoch. 
For each DPO training example $t=(x,y_w,y_l)$, we first compute its preference intensity $\psi$ (line~\ref{algline:compute-psi}). 
This step requires only a forward pass and is already performed as part of the DPO objective computation; thus, retrieving $\psi$ only introduces negligible overhead. In our implementation, we detach the preference intensity $\psi$ from the computational graph using a stop-gradient operator ($\texttt{sg}[\cdot]$ in line~\ref{algline:compute-psi}) to prevent backpropagation through the randomized preference intensity. 
Next, we sample a random rescaling weight $\tilde{w}$ according to the perturbation mechanism described in line~\ref{algline:random-weight}, and use $\tilde{w}$ to rescale the DPO objective before backpropagation.
The remainder of training proceeds identically to standard DPO. 

Intuitively, \privdpo{} pushes the model toward the true objective more aggressively than a standard DPO update with probability greater than $1/2$,  while with a slightly lower probability it moves in the opposite direction, thereby effectively confusing the potential adversary. 
Notably, this randomized rescaling process can be performed in a batched manner during the forward pass, making \privdpo{} a highly practical privacy‑preserving alignment method that is both efficient and scalable.

\vspace{1mm}
\noindent\textbf{Remark on noise injection.} 
In line~\ref{algline:obj-perturb} of Algorithm~\ref{alg:privdpo}, randomness is injected by rescaling the DPO objective, rather than by perturbing the raw preference label as in RR-based input perturbation. 
The perturbed objective is then passed directly to a standard optimizer, avoiding computationally and memory-intensive per-example gradient operations such as clipping or gradient perturbation.

\subsection{Privacy and Utility Analysis}\label{subsec:theoretical-analysis}
We theoretically analyze the privacy and utility guarantees of \privdpo{}, and present their practical implications. Complete proofs are deferred to Appendix \ref{appendix:pridpo-proof} and~\ref{appendix:proof-err-bound}.

\begin{thm}\label{thm:privacy-of-privdpo}
Algorithm \ref{alg:privdpo} satisfies $\epsilon$-preference privacy. 
\end{thm}

\begin{proof}[Proof sketch] Let $\mathcal{T}$ denote the space of DPO triplets. 
Consider an arbitrary DPO training record $t=(x,y_w,y_l)\in\mathcal{T}$, with $y_w \succ y_l$. 
Let $t'=(x,y_l,y_w)$ be its neighboring record, differing only in the preference signal, i.e., with $y_w \prec y_l$. 
Define the mechanism $\mathcal{M}: \mathcal{T}\mapsto \mathcal{R}$ which maps a DPO training example to its privatized gradient, corresponding to the randomized sub-process in lines~\ref{algline:compute-psi}-\ref{algline:get-gradient} of Algorithm~\ref{alg:privdpo}. 
By construction, \privdpo{} ensures that
\begin{align*}
    e^{-\epsilon}  \leq \frac{\Pr\left[\mathcal{M}(t)=r\right]}{\Pr\left[\mathcal{M}(t')=r\right]} \leq e^\epsilon
\end{align*}
for any $r\in \mathcal{R}$, thereby establishing the indistinguishability property of $\mathcal{M}$. 
By the transformation invariance property (Proposition~\ref{corollary:transformation-inv}), Algorithm~\ref{alg:privdpo} inherits this guarantee. The theorem follows. 
\end{proof}

\noindent\textbf{Remark.} Under the single-epoch protocol of Algorithm~\ref{alg:privdpo}, where the dataset is randomly permuted once and partitioned into disjoint mini-batches such that each example is used exactly once, Algorithm~\ref{alg:privdpo} satisfies dataset-level $\epsilon$-preference privacy.%

All training experiments in this paper use one epoch.
For neighboring datasets differing in one preference annotation, only the affected mini-batch differs; earlier releases are identical, and later updates are post-processing of the private release and unchanged examples.
Thus, no within-epoch composition is needed.
This differs from subsampling-based DP-SGD accounting~\cite{abadi2016deep,mironov2019r}, where examples may be selected repeatedly and privacy loss must be accumulated through subsampling and composition.

For multiple epochs, the same example may be accessed multiple times, and the privacy cost composes linearly by Proposition~\ref{prop:seq-composition}.
We use linear composition because our guarantee is pure $\epsilon$-preference privacy; tighter accounting for approximate or \Renyi{} DP~\cite{mironov2017renyi,mironov2019r} is not directly applicable in our setting.

Beyond satisfying $\epsilon$-preference privacy, \privdpo{} preserves key properties for maintaining alignment quality. The following lemma shows that the gradient derived from the randomized DPO objective is an unbiased estimator of the original DPO gradient. 
    
\begin{lemma}[Unbiasedness]\label{lemma:unbiasedness}
The gradient computed from the randomized DPO objective in line~\ref{algline:get-gradient} of Algorithm~\ref{alg:privdpo} is an unbiased estimator of the true DPO gradient:
\begin{align*}
    \mathbb{E}\left[ \nabla_\theta \tilde{\mathcal{L}}_t \right] = \nabla_\theta \mathcal{L}_{\rm DPO} (t;\pi_\theta,\pi_{\rm ref}).
\end{align*}
\end{lemma}
\begin{proof}[Proof sketch]
    Note that $\mathbb{E}\left[ \tilde{w} \mid \psi \right]=1$. Therefore, we have 
    \(\mathbb{E}[\nabla_\theta\tilde{\mathcal{L}}_t] = \mathbb{E}[{\tilde{w}}] \nabla_\theta\mathcal{L}_{\rm DPO}(t;\pi_\theta;\pi_{\rm ref})=\nabla_\theta\mathcal{L}_{\rm DPO}(t;\pi_\theta;\pi_{\rm ref}).\)
\end{proof}

The unbiasedness property implies that, in expectation, \privdpo{} performs LLM parameter updates identical to those of standard DPO. 
Therefore, with an appropriate batch size, the batch-averaged gradient (line~\ref{algline:agg-grad}) concentrates around the non-private DPO gradient. 
The following theorem formalizes this intuition with a concentration bound on the gradient deviation.

\begin{thm}[Error Bound of \privdpo{}]\label{thm:privdpo-grad-bound}
    Let $\boldsymbol{g}$ denote the batch-averaged gradient of \privdpo{} (line \ref{algline:agg-grad} in Algorithm \ref{alg:privdpo}) and $\boldsymbol{g}^*$ the corresponding non-private DPO gradient. 
    For the $i$-th training example $(x,y_w,y_l)$ in the batch, define $\boldsymbol{v}_i := \nabla_{\theta} \log\pi_{\theta}(y_w|x) - \nabla_{\theta} \log\pi_{\theta}(y_l|x)$.
    For any unit direction $\boldsymbol{u}\in\mathbb{S}^{d-1}$, define the directional deviation as 
    \[ {\rm Err}_{\boldsymbol{u}}:=\vert\boldsymbol{u}^\top (\boldsymbol{g}-\boldsymbol{g}^*)\vert. \] 
    Let $m$ be the batch size, then  with probability at least $1-\gamma$,
    \begin{align*}
        \mathrm{Err}_{\boldsymbol{u}} \leq &\frac{\beta}{m} \left( \sqrt{\frac{2 e^\epsilon}{(e^\epsilon-1)^2} \left( \sum_{i=1}^{m} (\boldsymbol{u}^\top \boldsymbol{v}_i)^2 \right) \log\frac{2}{\gamma}}\right.\\
     &\left.\quad\qquad\qquad +  \frac{e^\epsilon}{e^\epsilon-1} \cdot \frac{ \max_i\left\{\vert \boldsymbol{u}^\top \boldsymbol{v}_i \vert\right\} }{3} \log\frac{2}{\gamma} \right).
    \end{align*}
\end{thm}

\begin{proof}[Proof sketch]
For any $\boldsymbol{u}\in\mathbb{S}^{d-1}$, the per-example directional error $\boldsymbol{u}^\top (\boldsymbol{g}_i-\boldsymbol{g}_i^*)$ is a mean-zero random variable whose absolute value and variance are bounded. 
Applying a special case of Bernstein's inequality (Lemma~\ref{lemma:simplified-bern}) to the batch-averaged deviation $\mathrm{Err}_{\boldsymbol{u}}$ yields the stated bound.
\end{proof}

Theorem~\ref{thm:privdpo-grad-bound} guarantees that, with high probability, the randomized gradient in \privdpo{} remains close to the original DPO gradient along any fixed direction $\boldsymbol{u}\in\mathbb{S}^{d-1}$. 
Specifically, suppose  
\[\left\Vert\nabla_\theta \log \pi_{\theta} (y_w \mid x)-\nabla_\theta \log \pi_{\theta} (y_l \mid x)\right\Vert_2 \leq G\]
for some constant $G$.
Then for any $\boldsymbol{u}\in\mathbb{S}^{d-1}$ we have $|\boldsymbol{u}^\top \boldsymbol{v}_i| \le G$. 
Consequently, with probability at least $1-\gamma$, the deviation along $\boldsymbol{u}$ is
$\tilde{O}({\beta G}/{(e^\epsilon-1)\sqrt{m}})$, where $\tilde O(\cdot)$ suppresses logarithmic factors in $1/\gamma$. Moreover, since the update is unbiased (Lemma~\ref{lemma:unbiasedness}), Theorem~\ref{thm:privdpo-grad-bound} further implies that the randomized \privdpo{} update concentrates around the true DPO gradient along any direction.

We next compare \privdpo{} with two first-cut mechanisms. The most direct baseline under preference privacy is RR-based input perturbation, which perturbs the raw preference label before training. We therefore first analyze RR, and then discuss DP-SGD as a stronger but more conservative baseline.

Let $\boldsymbol{g}$ and $\boldsymbol{g}^*$ denote the batch-averaged \privdpo{} and non-private DPO gradients, respectively.
Adapting the proof of Theorem~\ref{thm:privdpo-grad-bound} in Appendix \ref{appendix:proof-err-bound} and applying a vector Bernstein inequality~\cite{gross2011recovering}, one can derive a high-probability upper bound of the form 

\begin{align}
\hspace{-2mm}\Vert\boldsymbol{g}-\boldsymbol{g^*}\Vert_2&=O\left(\frac{\beta}{m}\sqrt{\frac{e^\epsilon}{(e^\epsilon-1)^2} \lambda_{\max} \left( \sum_{t=1}^{m} \boldsymbol{v}_t\boldsymbol{v}_t^\top\right)}\right)=O\left(\frac{\beta}{\epsilon} \sqrt{\frac{
    \lambda}{m}}\right),\label{eq:grad-err-privdpo}
\end{align}
where $\lambda_{\max}(\cdot)$ and $\lambda$ denote the largest eigenvalue of a matrix and  $\mathbb{E}[\boldsymbol{v}_t\boldsymbol{v}_t^\top]$, respectively. 
The second equality follows from the approximation $\sum_{t=1}^{m}\boldsymbol{v}_t\boldsymbol{v}_t^\top \approx m \cdot \mathbb{E}[\boldsymbol{v}_t\boldsymbol{v}_t^\top]$, which holds when the empirical covariance concentrates, and on the fact that $\sqrt{e^\epsilon}/(e^\epsilon-1)=O(1/\epsilon)$ for small $\epsilon$. 

Together with the unbiasedness in Lemma~\ref{lemma:unbiasedness}, Eq.~\eqref{eq:grad-err-privdpo} suggests that the PrivDPO update can remain close to the non-private DPO update under low-effective-rank preference-axis covariance, which helps explain its stable empirical behavior. 

\vspace{1mm}
\noindent\textbf{Comparison with RR.} 
The most direct baseline is RR-based input perturbation, which keeps the true preference with probability $p=e^\epsilon/(e^\epsilon+1)$ and flips it with probability $q=1/(e^\epsilon+1)$. For a training example with non-private DPO gradient $-\beta\psi_i\boldsymbol{v}_i$, RR produces $-\beta\psi_i\boldsymbol{v}_i$ with probability $p$ and $\beta(1-\psi_i)\boldsymbol{v}_i$ with probability $q$. 
Applying an analysis similar to Eq.~\eqref{eq:grad-err-privdpo}, we obtain the batch-level deviation includes an additional generally nonzero bias term:
\[
\|\boldsymbol{g}^{\mathrm{RR}}-\boldsymbol{g}^*\|_2
=
O\left(
\frac{\beta}{\epsilon}\sqrt{\frac{\lambda}{m}}
+
\frac{\beta}{e^\epsilon+1}\sqrt{\lambda}
\right),
\]
where $\boldsymbol{g}^{\mathrm{RR}}$ denotes the batch gradient obtained under RR-perturbed preferences and $\boldsymbol{g}^*$ is the non-private DPO gradient. 
More explicitly, if $B_i \sim \mathrm{Bernoulli}(q)$ indicates whether the $i$-th preference is flipped, then
$g^{\mathrm{RR}}-g^*=\frac{\beta}{m}\sum_i(B_i-q)\boldsymbol v_i + \beta q\bar{\boldsymbol v}$, where
$\bar{\boldsymbol v}=\frac{1}{m}\sum_i\boldsymbol v_i$.
The first term is a zero-mean stochastic deviation, while the second term is the bias introduced by input-level RR.

In contrast, \privdpo{} is unbiased, and its stochastic deviation decreases with the batch size. This explains why perturbing the intermediate preference-dependent scalar is preferable to perturbing the raw preference label before DPO optimization.

We note that prior work on label DP~\cite{esmaeili2021antipodes,ghazi2021deep,busa2023label,esfandiari2022label,jiang2024protecting} focuses on supervised learning and is not directly applicable to our setting. These methods can be viewed as extensions of standard RR to deep learning settings, but they rely on assumptions about task structures that do not hold in LLM alignment. 
We therefore use RR as the canonical input-perturbation baseline and discuss these related label-DP methods in Section~\ref{subsec:threat-model} and Section~\ref{sec:related-work}.

\vspace{1mm}
\noindent\textbf{Comparison with DP-SGD.} Let $\tilde{\boldsymbol{g}}$ denote the privatized gradient obtained by DP-SGD, where isotropic Gaussian noise $\mathcal{N}(0,\sigma^2 \mathbb{I}^d)$ with $\sigma=O({\sqrt{\log(1/\delta)}}/{\epsilon})$ is added to the non-private DPO gradient~\cite{abadi2016deep,balle2018improving,mironov2019r}. Applying a standard sub-Gaussian tail bound yields
\begin{align}
    \Vert\tilde{\boldsymbol{g}}-\boldsymbol{g^*}\Vert_2=O\left( \frac{1}{\epsilon} \cdot \frac{\sqrt{d \log(1/\delta)}}{m} \right). \label{eq:grad-err-dpsgd}
\end{align}

Several properties of DPO-based LLM alignment help explain when this comparison is favorable to \privdpo{}. In Eq.~\eqref{eq:grad-err-privdpo}, the quantity $\lambda=\lambda_{\max}(\Sigma)$, where $\Sigma:=\mathbb{E}[\boldsymbol{v}_t\boldsymbol{v}_t^\top]$, reflects both directional concentration and gradient magnitude. Therefore, the comparison with DP-SGD in Eq.~\eqref{eq:grad-err-dpsgd} is most informative when $\Sigma$ has low effective rank and its leading eigenvalues remain bounded.
This condition is consistent with prior observations that LLM fine-tuning gradients often concentrate in low-dimensional intrinsic subspaces~\cite{aghajanyan2021intrinsic,li2018measuring,hulora,ding2023parameter,malladi2023fine,zhang2024dpzero}. We further verify this behavior empirically on Qwen2.5-3B-Instruct, initialized from the PrivSFT checkpoint described in Section~\ref{subsec:privsft}, over Anthropic-HH~\cite{bai2022training} using 32 gradient samples. The top 10 and top 20 components explain 62.5\% and 89.2\% of the empirical spectral energy, respectively, and $\lambda_{\max}/d=1.5\times10^{-4}$. Under such low-effective-rank behavior, $\lambda$ is much smaller than the dimension-dependent factor in Eq.~\eqref{eq:grad-err-dpsgd}.
This measurement is intended as representative empirical support for the low-effective-rank behavior observed in our setting, rather than a universal spectral guarantee across all models and datasets. 

\subsection{Private Supervised Fine-Tuning}\label{subsec:privsft}
In practice, DPO is preceded by a supervised fine-tuning (SFT) stage to improve training stability~\cite{rafailov2023direct}, typically by minimizing the negative log-likelihood of the preferred response $y_w$ given the prompt $x$, i.e., $-\log \pi_{\theta}(y_w\mid x)$. 
This, however, violates preference privacy, as it optimizes exclusively on preferred responses.  

To address this issue, we propose a modification to the SFT stage. 
Our key insight is that the role of SFT prior to DPO is not to directly reinforce preferences, but to adapt the LLM to the context of the alignment dataset~\cite{rafailov2023direct,ethayarajh2024model}. This adaptation allows the model to learn the dataset's semantics and the expected formatting of responses, {without incurring any privacy cost regarding user preference}. 

We use a private supervised fine-tuning method, termed as PrivSFT, that fine-tunes equally on both the preferred $y_w$ and the non-preferred $y_l$ responses for each prompt $x$, preventing preference leakage in SFT. The resulting PrivSFT loss is defined as:
\begin{align*}
    \mathcal{L}_{\rm PrivSFT}= - \frac{1}{2} \Big[\log \pi_{\theta} (y_w \mid x) + \log \pi_{\theta} (y_l \mid x) \Big].
\end{align*}
There is no privacy cost incurred in the private SFT stage, since $y_w$ and $y_l$ are public and they are symmetric in the above loss function (no human preference indicated). In our experiments, we first run one epoch of PrivSFT on the base model to obtain the reference model $\pi_\mathrm{ref}$, and then apply DPO-based alignment methods on $\pi_\mathrm{ref}$.

\section{Experiments}
\label{sec:experiments}
This section presents a comprehensive empirical evaluation of \privdpo{}. 
We first introduce the experimental setup in Section~\ref{subsec:experiment-setting}, then evaluate different aspects of alignment performance of \privdpo{} and its competitors from Section~\ref{subsec:exp-dpsgd} to Section~\ref{subsec:exp-winrate}, assess the robustness of \privdpo{} against an empirical attack in Section~\ref{subsec:mem-gap}, and finally validate the scalability of \privdpo{} in Section~\ref{subsec:exp-scalability}.

\subsection{Setup}\label{subsec:experiment-setting}
All experiments were conducted on a machine with 8$\times$ H20 GPUs, each with 96~GB of GPU memory.
We used PyTorch's Fully Sharded Data Parallel (FSDP)~\cite{zhao2023pytorch} to shard LLM parameters for efficient large-scale training. 
The implementation of \privdpo{} is available at: \url{https://github.com/Yangfan-Jiang/privatedpo}.

\vspace{1mm}
\noindent\textbf{Datasets.}
We use three well-established benchmarks in our experiments, all of which are open-ended text generation tasks. 
\begin{itemize}[leftmargin=*]
    \item \emph{Anthropic-HH}~\cite{bai2022training} is a widely used dataset for LLM alignment, collected by Anthropic to train LLMs that are both helpful and harmless. 
    It contains 161k human-LLM dialogue examples, each consisting of a dialogue context followed by two responses and a human preference label indicating the preferred one.
    \item \emph{TL;DR Summarization}~\cite{stiennon2020learning} contains 92k pairs of summaries for Reddit posts, with human preferences collected by OpenAI. Each example includes a Reddit post and an instruction to summarize it, along with two candidate summaries. Annotators select the summary that better captures the main points of the post. 
    \item \emph{UltraFeedback-Binarized}~\cite{cui2024ultrafeedback} contains 64k dialogue examples annotated for instruction following, truthfulness, honesty, and helpfulness.  
    Each example consists of a dialogue context and two responses, with a binary label indicating the better response. 
\end{itemize}
All datasets provide standard train-test splits; we use the training sets for training and the test sets for evaluating performance.

\vspace{1mm}
\noindent\textbf{Models.}  
We use models from three representative open-weight LLM families: Qwen2.5~\cite{qwen2024qwen25technicalreport}, Llama3~\cite{dubey2024llama}, and Pythia~\cite{biderman2023pythia}, ranging from 3B to 32B parameters. All models are publicly available on the Hugging Face platform.

\vspace{1mm}
\noindent\textbf{Hyperparameters.} We follow the hyperparameter settings from the original DPO paper~\cite{rafailov2023direct}. For all experiments, the learning rate is set to $5\times 10^{-7}$, the DPO hyperparameter $\beta$ to $0.1$, and the batch size to $64$, with 150 warm-up steps using a linear learning rate schedule. We do not tune the hyperparameters to favor any particular experimental setting. All models are trained using AdamW with the default settings provided by PyTorch. 
All experiments use the same single-epoch protocol as Algorithm~\ref{alg:privdpo}, i.e., one random permutation of the dataset followed by disjoint mini-batch processing, which is also consistent with~\cite{rafailov2023direct}.

\vspace{1mm}
\noindent\textbf{Metrics.} We evaluate alignment performance using two categories of metrics, following standard practices in prior DPO-based LLM alignment work~\cite{rafailov2023direct,meng2024simpo}:
\begin{itemize}[leftmargin=*]
    \item {Reward metrics}, including (i) {reward margin} (RM), the average of $\beta \log\frac{\pi_{\theta}(y_w\mid x)}{\pi_\mathrm{ref}(y_w\mid x)} - \beta\log\frac{\pi_{\theta}(y_l\mid x)}{\pi_\mathrm{ref}(y_l\mid x)}$ on the test set; (ii) {reward accuracy} (RA), the proportion of examples where $\log\frac{\pi_{\theta}(y_w\mid x)}{\pi_\mathrm{ref}(y_w\mid x)} > \log\frac{\pi_{\theta}(y_l\mid x)}{\pi_\mathrm{ref}(y_l\mid x)}$; and (iii) the DPO objective value $\mathcal{L}_{\rm DPO}$.
    \item {Generation quality metrics}, which compare the overall alignment quality of LLM-generated responses using an LLM-as-a-judge evaluation pipeline.
\end{itemize}

\vspace{1mm}
\noindent\textbf{Competitors.}
Existing privacy-preserving RLHF methods~\cite{chen2025improved,yu2024privacy,wu2023privately,he2024sample,zhang2025kl,zhang2025towards} typically protect training data with standard DP, effectively reducing to DP-SGD-style training. We therefore include DP-SGD as a standard DP baseline.

Under preference privacy, we compare with randomized response (RR), which randomly flips preference labels and gives biased estimators, and two stronger unbiased scalar perturbation baselines that replace PrivDPO's perturbation mechanism with Duchi's and Piecewise mechanisms~\cite{duchi2018minimax,wang2019collecting}.
These baselines build on our sensitivity analysis but do not exploit the structure of preference intensity identified in our analysis; details are given in Appendix~\ref{appendix:variant-privdpo}.

\vspace{1mm}
\noindent\textbf{Remark on privacy guarantee.} 
Our privacy guarantee applies to the preference annotations used in the DPO training procedure. It does not protect or remove information that may already be encoded in the pretrained or instruction-tuned checkpoint before our training begins. Since all methods in our experiments use the same initialization, datasets, and training pipeline, this limitation does not affect the fairness of the empirical comparison. 
In addition, because the instruction-tuning mixtures of some open-weight models are not fully disclosed, we cannot rule out partial overlap with public alignment datasets. Our evaluation should therefore be interpreted as comparing private DPO optimization methods under a common initialization, rather than auditing whether the base checkpoint has memorized these datasets.

\begin{figure*}[!t]
\centering
\begin{tabular}{ccc}
\multicolumn{3}{c}{\hspace{12mm} \includegraphics[height=7.0mm]{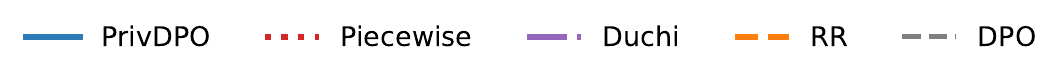}}\vspace{-5mm} \\
\hspace{-0mm}\subfigure[Reward Margin ($\uparrow$)] {\includegraphics[height=0.15\linewidth]{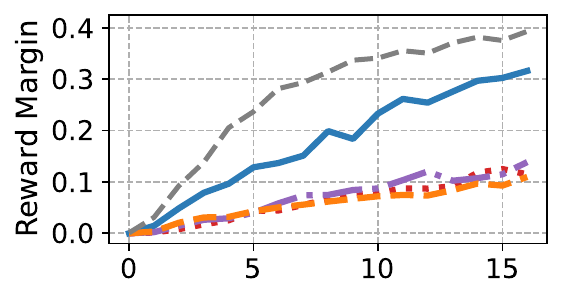}\label{subfig:rward-margin}} &
\hspace{-0mm}\subfigure[Reward Accuracy ($\uparrow$)] {\includegraphics[height=0.15\linewidth]{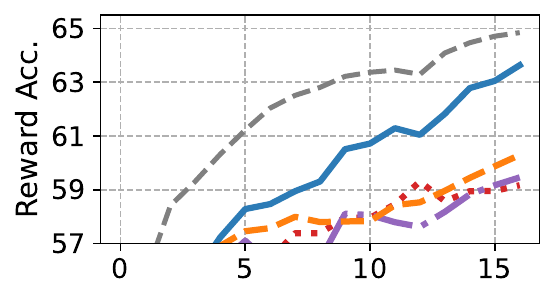}\label{subfig:reward-acc}} &
\hspace{-0mm}\subfigure[DPO Objective ($\downarrow$)]{\includegraphics[height=0.15\linewidth]{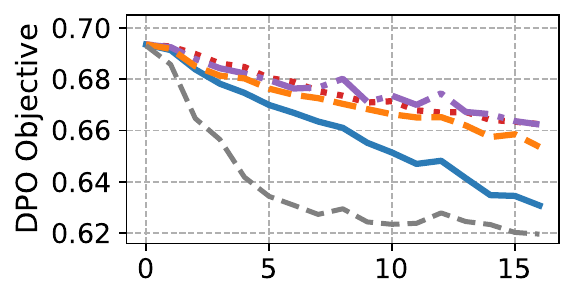}\label{subfig:eval-loss}} 
\end{tabular}
\vspace{-3mm}
\caption{Training dynamics on Anthropic-HH for Qwen2.5-3B-Instruct with privacy budget $\epsilon=1$. The x-axis denotes the number of training samples processed ($\times 10^4$).}
\label{fig:train-dynamics}
\vspace{-1mm}
\end{figure*}

\subsection{Limitations of Existing DP Mechanisms}\label{subsec:exp-dpsgd}
We apply DP-SGD~\cite{abadi2016deep} to DPO training on Qwen2.5-3B-Instruct with the Anthropic-HH dataset, enforcing $(\varepsilon=16,\delta={1}/{10n})$-DP, where $n$ is the number of training examples. 
Following prior work on DP fine-tuning of LLMs~\cite{bao2025unlocking,lilarge,yudifferentially}, we set the gradient clipping threshold to $1$; further details are provided in Appendix~\ref{appendix:exp-dp-sgd}.

After one epoch, the test set DPO objective is $0.688$ (vs. $0.693$ before training and $0.62$ for non-private DPO). 
The reward margin drops to $0.01$ (vs. $0.39$ non-private), and reward accuracy decreases to $55\%$ (vs. $64\%$ non-private). 
Varying the clipping threshold in $\{0.1, 1, 10, 100\}$ yields no clear improvement, indicating that DP-SGD fails to capture meaningful preference signals even at the relatively weak privacy level $\varepsilon=16$. 

Besides undesirable alignment performance, DP-SGD suffers from several limitations that significantly limit its practical deployment for large-scale LLM alignment. 

\vspace{1mm}
\noindent\textbf{Memory barrier.}
Existing large-scale training frameworks, such as FSDP~\cite{zhao2023pytorch}, shard parameters, gradients, and optimizer states across GPUs, which is incompatible with per-example gradient clipping.  
Supporting per-example clipping requires pure data parallelism, forcing each GPU to hold full model weights, gradients, and optimizer states. 
For Qwen2.5-3B, this already exceeds $80$~GB per GPU, while larger models trigger out-of-memory failures on our hardware. 
Similar bottlenecks have been reported in prior work on private LLM fine-tuning~\cite{bao2025unlocking,zhang2024dpzero,tang2025private}.

\vspace{1mm}
\noindent\textbf{Low efficiency.}  
DP-SGD also suffers from poor computational efficiency.  
Per-example gradient operation prevents effective batching, substantially increasing training time. 
On our hardware, DP-SGD training for Qwen2.5-3B takes $6\times$ longer than standard DPO.

Similar results are observed on Llama3.2-3B-Instruct and Pythia-2.8B models, indicating that off-the-shelf DP-SGD is unsuitable for aligning LLMs under current privacy accounting and optimization techniques.  
Our subsequent evaluations therefore compare \privdpo{} mainly with stronger baselines, specifically, variants of \privdpo{} that also satisfy $\epsilon$-preference privacy.


\begin{table*}[!t]
\centering
\caption{Performance comparison on Anthropic-HH dataset}
\vspace{-3mm}
\label{tb:main-hh}
\begin{small}
\begin{tabular}{llccccccccc}
\toprule
\multirow{2}{*}{\textbf{Privacy Budget}} & \multirow{2}{*}{\textbf{Method}} & \multicolumn{3}{c}{\textbf{Pythia-2.8B}}                 & \multicolumn{3}{c}{\textbf{Qwen2.5-3B-Instruct}} & \multicolumn{3}{c}{\textbf{Llama3.2-3B-Instruct}} \\ 
\cmidrule(lr){3-5} \cmidrule(lr){6-8} \cmidrule(lr){9-11}

                         &                          & \textbf{RM} ($\uparrow$)  & \textbf{RA} ($\uparrow$)  & $\mathcal{L}_{\rm DPO}$ ($\downarrow$)   & \textbf{RM} ($\uparrow$)   & \textbf{RA} ($\uparrow$)  & $\mathcal{L}_{\rm DPO}$ ($\downarrow$)  & \textbf{RM} ($\uparrow$)  & \textbf{RA} ($\uparrow$)   & $\mathcal{L}_{\rm DPO}$ ($\downarrow$)   \\ \midrule
\multirow{4}{*}{$\epsilon=0.5$}  & RR          & 0.026  & 54.6\%  & 0.69  &  0.040  &  58.6\%   &  0.68  &  0.045  &  59.0\%  &  0.68 \\
                         & Duchi               & 0.068  & 55.3\%  & 0.69  &  0.083  &  56.9\%   &  0.68  &  0.092  &  56.6\%  &  0.68 \\
                         & Piecewise           & 0.038  & 53.6\%  & 0.69  &  0.041  &  56.1\%   &  0.68  &  0.068  &  56.5\%  &  0.68 \\
                         & \textbf{\privdpo{}} & \textbf{0.108} & \textbf{56.9\%} & \textbf{0.67} & \textbf{0.154} & \textbf{59.4\%} &  \textbf{0.65}  &  \textbf{0.186}  &  \textbf{61.1\%}  & \textbf{0.65} \\ \midrule
\multirow{4}{*}{$\epsilon=1$}   & RR           & 0.063  & 58.8\%  & 0.67   &  0.109  &  60.3\%  &  0.65 &   0.143  &  62.7\%  &  0.64  \\
                         & Duchi               & 0.103  & 56.7\%  & 0.68   &  0.138  &  59.5\%  &  0.66 &   0.172  &  60.9\%  &  0.66  \\
                         & Piecewise           & 0.082  & 56.1\%  & 0.68   &  0.115  &  59.2\%  &  0.66 &   0.145  &  60.6\%  &  0.65  \\
                         & \textbf{\privdpo{}} & \textbf{0.156} & \textbf{59.8\%} & \textbf{0.65} & \textbf{0.316} & \textbf{63.6\%} & \textbf{0.63} & \textbf{0.359} &  \textbf{64.9\%}   &  \textbf{0.62}  \\ \midrule
Non Private              & DPO                 & 0.309  & 62.5\%  & 0.63  &  0.393 & 64.8\% &  0.62 &  0.487 & 66.2\% &  0.60 \\ \bottomrule
\end{tabular}
\end{small}
\end{table*}

\subsection{Main Results}\label{subsec:exp-main}
We compare the alignment performance of \privdpo{} against its privacy-preserving competitors, as well as against the standard DPO approach (with standard, non-private SFT) as presented in~\cite{rafailov2023direct}, to evaluate the reward gap between private and standard versions. 
In each experiment, following common practice in standard DPO training~\cite{rafailov2023direct,meng2024simpo}, we train the LLM for one epoch on the training set and evaluate the reward metrics on the test set. 
Due to the scale of the experiments, we mainly evaluate performance using a single random seed, fixed to $0$ across all runs. Preliminary experiments with additional seeds show low variability, with standard deviation around $1\%$ for RA and $\mathcal{L}_{\rm DPO}$, and less than $5\%$ for RM, and are therefore omitted for brevity. 
We report results under two representative preference privacy budgets, $\epsilon=0.5$ and $\epsilon=1$, which provide strong and practically meaningful privacy guarantees. 

Figure~\ref{fig:train-dynamics} illustrates the training dynamics on the Anthropic-HH dataset using Qwen2.5-3B-Instruct under a preference privacy budget of $\epsilon=1$. We observe that \privdpo{} achieves convergence comparable to the non-private DPO method across all three metrics, while significantly outperforming the privacy-preserving baselines. 

The evaluation results on all three benchmarks are reported in Tables~\ref{tb:main-hh}, \ref{tb:main-tldr} (in Appendix~\ref{appendix:tables}), and \ref{tb:main-shp} (in Appendix~\ref{appendix:tables}), respectively. All three sets of results show similar trends, from which we make the following observations. 
{First}, \privdpo{} consistently and significantly outperforms all privacy-preserving competitors on all metrics, across different privacy budgets, LLMs, and benchmarks. 
Specifically, for the most important metric, reward margin, which reflects the degree to which the model prefers human-preferred responses over non-preferred ones, \privdpo{} typically achieves a $2$-$3\times$ larger margin. This indicates that \privdpo{} enables the model to move more effectively toward human-aligned behavior. We also observe clear improvements in reward accuracy and the DPO objective value, suggesting that \privdpo{} generally converges more stably and learns human preferences more accurately than competing methods. 
Second, \privdpo{} consistently outperforms both advanced baselines, Duchi and Piecewise, across all evaluated settings; meanwhile, all structured mechanisms substantially outperform the naive RR baseline, which suffers severe utility degradation due to biased perturbation. 
{Third}, the performance of \privdpo{} is comparable to that of standard DPO. For a moderate yet strong privacy budget of $\epsilon=1$, \privdpo{} achieves results close to the non-private DPO, suggesting that \privdpo{} preserves alignment utility while providing meaningful preference privacy protection. For a stricter privacy budget of $\epsilon=0.5$, the results remain competitive, although a small performance gap exists, representing the cost of privacy. Improving utility under such stringent privacy guarantees is an interesting direction for future work.

We further evaluate the privacy-utility trade-off on Qwen2.5-3B-Instruct over Anthropic-HH with $\epsilon\in\{0.5,1,1.5,2\}$, as shown in Figure~\ref{fig:vary-eps}. This experiment is limited to one representative model and dataset due to the cost of LLM training. The results show that \privdpo{} improves steadily as the privacy budget increases and remains consistently stronger than all baselines.

\begin{figure*}[!t]
\centering
\begin{tabular}{ccc}
\multicolumn{3}{c}{\hspace{12mm} \includegraphics[height=7.0mm]{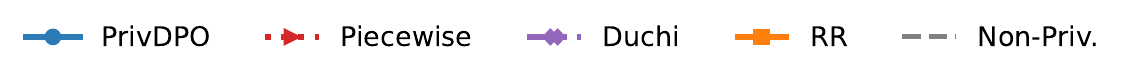}}\vspace{-5mm} \\
\hspace{-0mm}\subfigure[Reward Margin ($\uparrow$)] {\includegraphics[height=0.15\linewidth]{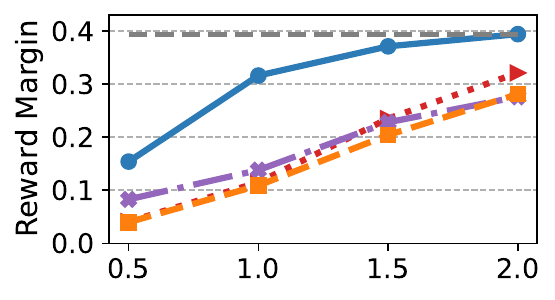}\label{subfig:rward-margin-eps}} &
\hspace{-0mm}\subfigure[Reward Accuracy ($\uparrow$)] {\includegraphics[height=0.15\linewidth]{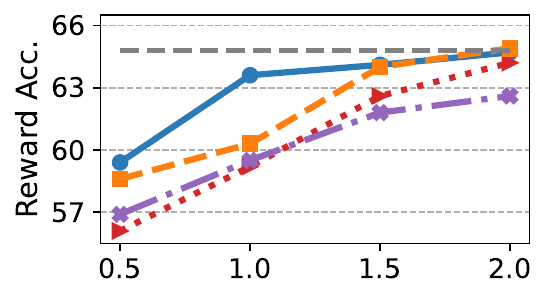}\label{subfig:reward-acc-eps}} &
\hspace{-0mm}\subfigure[DPO Objective ($\downarrow$)]{\includegraphics[height=0.15\linewidth]{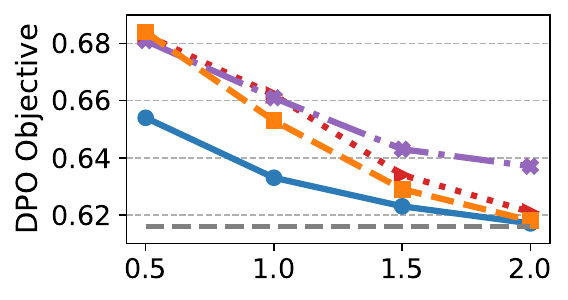}\label{subfig:eval-loss-eps}} 
\end{tabular}
\vspace{-3mm}
\caption{Varying privacy budgets ($\epsilon$) on Anthropic-HH for Qwen2.5-3B-Instruct. The x-axis denotes the privacy budget.}
\label{fig:vary-eps}

\vspace{-1mm}
\end{figure*}

\begin{table}[!t]
\centering
\caption{Win rate of \textbf{\privdpo{}} on Anthropic-HH (\%)}
\vspace{-3mm}
\label{tb:winrate-hh}
\begin{small}
\begin{tabular}{lccc}
\toprule
\textbf{\privdpo{}} vs. & \textbf{Pythia-2.8B} & \textbf{Qwen2.5-3B} & \textbf{Llama3.2-3B} \\ 
\midrule
PrivSFT       &  61.3   &  74.5    & 72.8    \\\midrule
RR            &  57.7   &  69.6    & 63.5    \\
Duchi         &  56.1   &  60.1    & 58.6    \\
Piecewise     &  58.3   &  66.7    & 58.1    \\\midrule
DPO           &  40.2   &  48.7    & 42.7     \\
\bottomrule
\end{tabular}
\vspace{-1mm}
\end{small}
\end{table}


\subsection{Comparison of Text Generation Quality}\label{subsec:exp-winrate}
We evaluate generation quality using LLM-as-a-judge~\cite{rafailov2023direct,zheng2023judging} win rates at $\epsilon=0.5$. 
Responses are sampled with the {vLLM} framework~\cite{kwon2023efficient} using temperature $0.7$ and repetition penalty $1.1$, and GPT-4.1 judges response pairs following the DPO evaluation pipeline~\cite{rafailov2023direct}. 
Each evaluation input consists of a prompt $x$ and a pair of responses $y$ and $y'$ generated by the LLMs trained with \privdpo{} and its competitors, respectively. 
The judge model determines which response is better according to task-specific criteria. For Anthropic-HH, the evaluation focuses on helpfulness and harmlessness; for TL;DR, on precision and conciseness; and for UltraFeedback, on instruction-following, truthfulness, honesty, and overall helpfulness.

The win rate comparisons on all three datasets are reported in Tables~\ref{tb:winrate-hh}, \ref{tb:winrate-tldr} (in Appendix~\ref{appendix:tables}), and \ref{tb:winrate-ultra} (in Appendix~\ref{appendix:tables}). The observed trends are consistent with the reward-based results presented in the previous section, from which we make the following observations. 
First, \privdpo{} significantly outperforms PrivSFT, indicating that it effectively learns human preferences from the DPO dataset while provably protecting individual preference information. Second, \privdpo{} demonstrates a clear advantage over all competitors, achieving win rate around or above $58\%$ in most cases. Finally, responses generated by LLMs trained by \privdpo{} are comparable to those from the standard DPO method, achieving competitive win rates above $47\%$ in most cases and exceeding $40\%$ across all cases. 

\subsection{Evaluation of Memorization Advantage}
\label{subsec:mem-gap}

\begin{table}[!t]
\centering
\caption{Empirical memorization advantage (\%)}
\vspace{-3mm}
\label{tb:memadv}
\begin{small}
\begin{tabular}{lccc}
\toprule
{Methods} & Anthropic-{HH} & {TL;DR}  & {UltraFeedback} \\ 
\midrule
\privdpo{}   &   0.01   &  0.51  &   0.85   \\
\midrule
RR    &   0.17   &  0.55  &   0.87   \\
Duchi     &   0.21   &  0.18  &   0.14   \\
Piecewise     &   0.15   &  0.39  &   0.06   \\
\midrule
DPO      &  9.76    &  7.36   &    9.58  \\
\bottomrule
\end{tabular}
\vspace{-2mm}
\end{small}
\end{table}

We include a simple empirical attack as a sanity check to evaluate whether alignment training (via DPO or \privdpo{}) makes preference labels more exploitable on training examples. 
The attacker predicts the preferred response by comparing model-derived scores for the two candidates; implementation details are in Appendix~\ref{appendix:attack}. 

Since preference signals may be partially inferred even without access to the data, raw attack accuracy is not interpreted as privacy leakage. 
Instead, we report the {memorization advantage} (denoted {MemAdv}) that isolates the attacker's \emph{additional advantage} on training pairs beyond prior and generalization. 
Let $\pi_{\theta}$ denote the aligned model and $\pi_\mathrm{prior}$ a prior model not trained on preference labels (the base pretrained model in our experiments). 
Let $A(\pi,S)$ be the attack success rate on split $S\in\{\mathrm{train},\mathrm{test}\}$. We define MemAdv as
\begin{align*}
(A(\pi_{\theta},\mathrm{train})-A(\pi_\mathrm{prior},\mathrm{train}))\;-\;(A(\pi_{\theta},\mathrm{test})-A(\pi_\mathrm{prior},\mathrm{test})).
\end{align*}

Intuitively, the first term captures additional preference information revealed on seen examples, while the second removes gains attributable to generalization. 
Thus, a larger MemAdv indicates disproportionate improvement on training pairs, 
consistent with memorization of private preference signals; a small MemAdv suggests limited additional advantage on training pairs beyond what is already achievable by the prior and by generalization.

Table~\ref{tb:memadv} reports MemAdv for \privdpo{} and its competitors under privacy budget $\epsilon=0.5$; the $\epsilon=1$ results are reported in Table~\ref{tb:memadv-appendix} in Appendix~\ref{appendix:attack}. 
We omit DP-SGD from the table because it is evaluated only as a standard-DP sanity baseline under a different privacy regime and is not competitive or scalable in our DPO setting. 
Standard DPO shows a clear memorization advantage, whereas all private methods yield substantially smaller values. \privdpo{} remains near zero across datasets and privacy budgets, suggesting limited additional leakage of training preference labels beyond what is already inferable from the prior model and generalization. RR exhibits similar behavior under the same preference-privacy guarantee, while Duchi and Piecewise sometimes produce smaller MemAdv because they inject stronger scalar perturbations than necessary for our preference-privacy notion.

\subsection{Scalability Analysis}\label{subsec:exp-scalability}
We evaluate scalability in terms of alignment performance and training efficiency, with all experiments conducted on the Anthropic-HH dataset with a fixed privacy budget of $\epsilon=0.5$. 
To isolate the effect of model size from confounding factors such as model architecture and pre-training data distribution, we perform the scalability analysis using models from the Qwen2.5 family, which share a similar architecture and pre-training corpus~\cite{qwen2024qwen25technicalreport}.

Results for reward metrics and win rates are reported in Table~\ref{tb:vary-llm-sizes} and Figure~\ref{subfig:win-rate-scalability}, respectively. 
The results show that \privdpo{} scales effectively to models with up to 32B parameters, consistently outperforming the representative RR baseline and achieving win rates above $70\%$ against the PrivSFT method, indicating strong alignment performance and meaningful preference learning.

For the training efficiency analysis, we compare the training time of \privdpo{} with standard DPO. 
Since \privdpo{} only perturbs the DPO objective, it can serve as a direct drop-in replacement within any accelerated LLM training framework, benefiting from the underlying optimizations. 
The results in Figure~\ref{subfig:training-time} demonstrate that \privdpo{} achieves preference privacy while introducing only a negligible additional cost in training time.

\begin{table}[!t]
\centering
\caption{Performance on Anthropic-HH across model sizes}
\vspace{-2mm}
\label{tb:vary-llm-sizes}
\begin{small}
\begin{tabular}{llccc}
\toprule
\textbf{Model Size} & \textbf{Method} & \textbf{RM} ($\uparrow$)  & \textbf{RA} ($\uparrow$)  & $\mathcal{L}_{\rm DPO}$ ($\downarrow$) \\ 
\midrule
\multirow{2}{*}{3B}      & RR         & 0.040   & 58.6\%   &  0.68 \\
     & \textbf{\privdpo{}}            & \textbf{0.154}   & \textbf{59.4\%}   &  \textbf{0.65} \\\midrule
\multirow{2}{*}{7B}      & RR         & 0.085   & 58.9\%   &  0.67 \\
     & \textbf{\privdpo{}}            & \textbf{0.383}   & \textbf{63.2\%}   &  \textbf{0.64} \\\midrule
\multirow{2}{*}{14B}      & RR        & 0.081   & 59.1\%   &  0.67 \\
     & \textbf{\privdpo{}}            & \textbf{0.388}   & \textbf{65.6\%}   &  \textbf{0.62} \\\midrule
\multirow{2}{*}{32B}      & RR        & 0.056   & 55.5\%   &  0.68 \\
     & \textbf{\privdpo{}}            & \textbf{0.358}   & \textbf{65.2\%}   &  \textbf{0.62} \\
\bottomrule
\end{tabular}
\vspace{-1mm}
\end{small}
\end{table}

\begin{figure}[!t]
\centering
{
\begin{tabular}{cc}
\hspace{-2mm}\subfigure[Win Rate of \privdpo{}] {\includegraphics[height=0.295\linewidth]{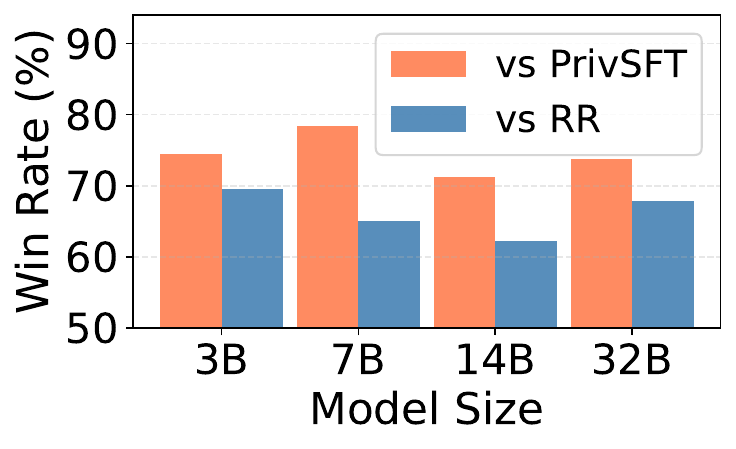}\label{subfig:win-rate-scalability}}  & 
\hspace{-2mm}\subfigure[Training Time (hours)] {\includegraphics[height=0.295\linewidth]{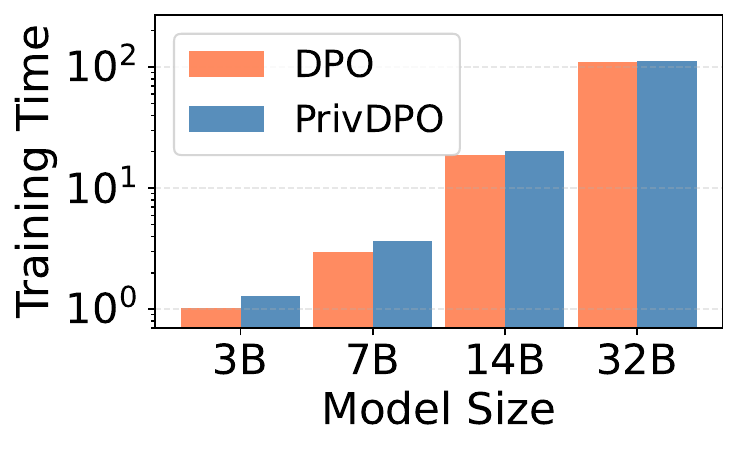}\label{subfig:training-time}}
\end{tabular}
}
\vspace{-2mm}
\caption{Win rate and training time across model sizes.}
\label{fig:scalability}

\vspace{-1mm}
\end{figure}

\section{Related Work}\label{sec:related-work}
\balance
Privacy risks in LLMs have been widely studied~\cite{tramer2024position}, and it is now well established that privacy leakage in LLMs is a practical concern. Such risks include membership inference and memorization~\cite{carlini2021extracting,hayes2025measuring,carlini2022quantifying}, as well as reconstruction attacks~\cite{carlini2019secret,nasr2025scalable}. Recent work has also highlighted privacy issues in LLM alignment~\cite{barbero2025extracting,feng2025exposing}. 
To mitigate these risks, privacy-preserving techniques for LLM training have been extensively studied. 
However, most existing approaches focus on supervised fine-tuning under DP guarantees~\cite{yudifferentially,bao2025unlocking,lilarge}, rather than on LLM alignment. 
Recent efforts have also investigated pre-training LLMs with DP~\cite{sinha2025vaultgemma}, but mainly focus on a relatively small-scale model with around 1B parameters. %

In the context of LLM alignment, only a few recent works study privacy-preserving methods, and most consider DP~\cite{chen2025improved,yu2024privacy,wu2023privately,goel2025differentially,hou2025private} or theoretical aspects~\cite{he2024sample,zhang2025kl,zhang2025towards}. 
Existing DP-based alignment methods~\cite{chen2025improved,yu2024privacy,wu2023privately} are typically applied to small models such as RoBERTa or GPT-2, with million-scale parameters. 
As our analysis and experiments show, these approaches become overly restrictive and impractical when extended to billion-parameter LLMs.  
Theoretical works focus on sample complexity~\cite{he2024sample}, lower bounds~\cite{zhang2025towards}, or convergence~\cite{zhang2025kl}, offering limited practical insights for large-scale LLM alignment. 
In addition, some studies aim to ensure privacy through pre-processing or post-processing rather than during alignment training. 
For example, \citet{yu2024privacy} generates synthetic alignment datasets with DP guarantees but only for specific instruction-following tasks; \citet{hou2025private} uses a DPO-based DP data generator in a federated setting, focusing on synthetic data generation rather than ensuring DP within DPO optimization; and \citet{goel2025differentially} applies model editing at inference time to achieve DP guarantees. 
These approaches are orthogonal to ours, which develops a fully end-to-end privacy-preserving optimization framework for LLM alignment. Overall, existing methods either target narrow task settings or focus on small-scale models, lack strong evidence for more general LLM alignment in practical settings. 

Several recent works study privacy in preference alignment under complementary assumptions. \citet{zhou2025square} formulate private preference optimization using the exponential mechanism, which provides a clean theoretical construction but is difficult to instantiate directly for billion-parameter LLM training. \citet{weng2025improved} analyze private and robust alignment under a finite hypothesis class assumption, where the output model is selected from a finite candidate set rather than optimized directly in a large parameter space. \citet{zhang2025towards} study user-level private PPO-style RLHF, which targets a different alignment paradigm from the DPO setting considered here. These works provide important theoretical insights, but are complementary to our goal of designing a scalable DPO optimizer under preference privacy.

Another closely related line studies preference-level privacy in alignment~\cite{teku2025props}, proposing a PATE-based framework~\cite{esmaeili2021antipodes,papernot2018scalable} combined with label-DP mechanisms~\cite{ghazi2021deep,busa2023label,esfandiari2022label}.  
However, their method requires multi-stage training that limits scalability. 
Moreover, their pipeline uses a subset of the training data for non-private SFT, which is inconsistent with our definition of preference privacy and does not ensure the same $\epsilon$-preference privacy guarantee.  
In contrast, our approach directly modifies the core DPO optimizer to enforce end-to-end preference privacy, enabling efficient training of production-scale LLMs on large real-world datasets.

Finally, several DPO variants have been proposed for LLM alignment, including KTO~\cite{ethayarajh2024model}, IPO~\cite{azar2024general}, and SimPO~\cite{meng2024simpo}. Our privacy analysis depends on the exact form of the standard DPO objective and thus does not directly extend to these variants. Developing privacy-preserving mechanisms for these DPO variants would be an interesting direction for future work.

\section{Conclusion}
\label{sec:conclusion}
This paper shows that aligning production-scale LLMs via DPO under strong preference privacy guarantees is practical. Across multiple alignment benchmarks and model families, we demonstrate that it is possible to achieve performance comparable to standard non-private DPO while enforcing a strict preference privacy constraint. These results are enabled by \privdpo{}, a simple yet effective method that injects carefully calibrated randomness into the DPO process to preserve utility while providing formal preference privacy. 
\privdpo{} represents a practical step toward deploying privacy-preserving RLHF methods in real-world systems. 

Several directions remain open. 
Future work includes extending \privdpo{} to other alignment objectives and post-training methods, moving from record-level to user-level preference privacy, exploring stronger privacy mechanisms and compositions, and adapting preference privacy to multi-turn or interactive feedback settings.

\clearpage

\bibliographystyle{ACM-Reference-Format}
\bibliography{reference}

\appendix
\section{First-Cut Solutions}\label{appendix:first-cut-solution}
We present two DP-based first-cut solutions that can be adapted for aligning LLMs using DPO with preference privacy. 
However, as discussed below, these approaches have fundamental technical limitations, as they were not specifically designed for DPO tasks under preference privacy constraints. 

\vspace{1mm}
\noindent\textbf{DP-SGD.} 
The DP-SGD algorithm~\cite{abadi2016deep} achieves DP by injecting noise to gradients during training, and has been widely applied in deep learning tasks. 
In the context of LLMs, prior work~\cite{bao2025unlocking,yudifferentially,lilarge} has mainly used DP-SGD to fine-tune small models (under 1B parameters, e.g., RoBERTa, GPT-2) on relatively simple datasets (e.g., sentiment classification), which are less representative of production-scale LLM training for open-ended text generation. 
When applied to practical DPO-based LLM alignment, DP-SGD incurs substantial performance degradation, as shown in our theoretical analysis in Section~\ref{subsec:theoretical-analysis} and empirical evaluation in Section~\ref{subsec:exp-dpsgd}. 
This is largely due to DP's pessimistic assumption about the adversary, which enforces an overly strict privacy guarantee and require injecting large noise uniformly across all parameters. 
Such noise leads to a prohibitive $\Omega(\sqrt{d})$ error in the released model parameters~\cite{dinur2003revealing}. %

\vspace{1mm}
\noindent\textbf{Randomized preference flipping.}
This approach is a label-DP-based solution that traces back to the classical randomized response (RR) technique~\cite{warner1965randomized}. 
For simple counting problems, RR can yield unbiased estimates through post-processing~\cite{wang2017locally}. However, in the DPO setting, the final information revealed to a potential adversary is not a simple aggregate count, but rather the LLM parameters (or gradients) that encode preference information in a highly complex manner. 
Naively applying input-level RR followed by standard DPO produces biased gradients. Removing this bias requires exploiting the DPO-specific scalar structure, as \privdpo{} does, rather than a simple post-processing of the flipped preferences.

\section{Proofs}

\allowdisplaybreaks
\subsection{Proof of Theorem \ref{thm:privacy-of-privdpo}}\label{appendix:pridpo-proof}
\begin{proof}%
Denote by $\mathcal{T}$ the space of DPO triplets. 
Consider an arbitrary DPO training example $t=(x,y_w,y_l)\in\mathcal{T}$, and let $t'=(x,y_l,y_w)$ be the preference neighbor of $t$. We first establish the indistinguishability between the gradients computed on $t$ and $t'$.

Let $\mathcal{M}: \mathcal{T}\mapsto \mathcal{R}$ denote the mechanism that maps an input DPO training example to the privatized gradient, corresponding to the randomized sub-process in lines~\ref{algline:compute-psi}-\ref{algline:get-gradient} of Algorithm~\ref{alg:privdpo}. 
Denote by 
$$\boldsymbol{v}:=\nabla_\theta \log \pi_{\theta} (y_w \mid x)-\nabla_\theta \log \pi_{\theta} (y_l \mid x)$$ 
the preference axis associated with the pair of neighboring DPO training records $t$ and $t'$.  
We partition the output range $\mathcal{R}:=\mathbb{R}^d$ into three disjoint subsets as follows:
\begin{align*}
    r_1& :=\left\{-\beta \left(\psi-1-\frac{1}{e^\epsilon-1}\right) \boldsymbol{v}\right\}, \\
    r_2& :=\left\{-\beta\left(\psi+1-\frac{e^\epsilon-2}{e^\epsilon-1}\right) \boldsymbol{v}\right\}, \\
    r_3& :=\mathcal{R}\setminus \left( r_1 \cup r_2 \right).
\end{align*}

We now consider the above output space case-by-case to bound the pdf ratio between $\mathcal{M}(t)$ and $\mathcal{M}(t')$. 

\vspace{1mm}
\noindent\textbf{Case 1: $\mathcal{M}(\cdot)\in r_1$.} For the input DPO record $t$, we have
\begin{align*}
    \nabla_{\theta} \mathcal{L}_{\rm DPO} (t) := \nabla_{\theta} \mathcal{L}_{\rm DPO} (t;\pi_{\theta};\pi_{\rm ref}) = -\beta \psi \boldsymbol{v}.
\end{align*}
The probability of $\mathcal{M}(t) \in r_1$ is:
\begin{align}
     &\Pr\left[ \mathcal{M}(t)=-\beta \left( \psi -1 -\frac{1}{e^\epsilon-1} \right) \boldsymbol{v} \right] \nonumber\\
    = & \Pr\left[ \mathcal{M}(t)= \left( 1 - \frac{1}{\psi} ( 1 +\frac{1}{e^\epsilon-1}) \right) \left(-\beta \psi\boldsymbol{v}\right) \right] \nonumber\\
    = & \Pr\left[ \mathcal{M}(t)= \left( 1 - \frac{1}{\psi} ( 1 +\frac{1}{e^\epsilon-1}) \right) \nabla_{\theta} \mathcal{L}_{\rm DPO} (t) \right] \nonumber\\
    = & \Pr\left[\tilde{w}= 1 - \frac{1}{\psi} ( 1 +\frac{1}{e^\epsilon-1}) \right] = \frac{1}{e^\epsilon+1}.\label{eq:g=r1}
\end{align}

For the neighboring DPO record $t'$, we have:
\begin{align*}
    \psi'=\sigma \left(\beta\log\frac{\pi_{\theta}(y_w\mid x)}{\pi_{\rm ref}(y_w \mid x)} - \beta\log\frac{\pi_{\theta}(y_l\mid x)}{\pi_{\rm ref}(y_l \mid x)}\right)=1-\psi,
\end{align*}
and thus the randomized rescaling weight for input $t'$ is:
\begin{align*}
\tilde{w}' \gets 
    \left\{
    \begin{array}{ll}
       1+\frac{1}{1-\psi}\left(1-\frac{e^\epsilon-2}{e^\epsilon-1}\right), & \text{if } u\leq \frac{e^\epsilon}{e^\epsilon+1},\\
      1-\frac{1}{1-\psi}\left(1+\frac{1}{e^\epsilon-1}\right), & \text{otherwise,}
    \end{array}
    \right.
\end{align*}
where $u$ is sampled uniformly at random from $[0,1]$. 
Note that for $t'=(x,y_l,y_w)$, we have
\begin{align*}
    \nabla_{\theta} \mathcal{L}_{\rm DPO} (t';\pi_{\theta};\pi_{\rm ref}) = -\beta\psi'(-\boldsymbol{v}) = -\beta (1-\psi) (-\boldsymbol{v}).
\end{align*}
Thus, the probability that $\mathcal{M}(t') \in r_1$ is:
\begin{align}
     &\Pr\left[ \mathcal{M}(t')=-\beta \left( \psi -1 -\frac{1}{e^\epsilon-1} \right) \boldsymbol{v} \right] \nonumber\\
    = & \Pr\left[ \mathcal{M}(t')=-\beta \left(  1 - \psi + 1 - \frac{e^\epsilon-2}{e^\epsilon-1} \right) (-\boldsymbol{v}) \right] \nonumber\\
    = & \Pr\left[ \mathcal{M}(t')= \left(  1 + \frac{1}{\psi'} (1-\frac{e^\epsilon-2}{e^\epsilon-1})  \right) (-\beta) (-\boldsymbol{v})\psi' \right] \nonumber\\
    = & \Pr\left[ \mathcal{M}(t')= \left(  1 + \frac{1}{\psi'} (1-\frac{e^\epsilon-2}{e^\epsilon-1}) \right) \nabla_{\theta} \mathcal{L}_{\rm DPO} (t') \right] \nonumber\\
    = & \Pr\left[\tilde{w}'= 1 + \frac{1}{1-\psi} ( 1 - \frac{e^\epsilon-2}{e^\epsilon-1}) \right] = \frac{e^\epsilon}{e^\epsilon+1}.\label{eq:g'=r1}
\end{align}

\vspace{1mm}
\noindent\textbf{Case 2: $\mathcal{M}(\cdot)\in r_2$.} For the input DPO triple $t$, we have
\begin{align}
    & \Pr\left[ \mathcal{M}(t)=-\beta\left(\psi+1-\frac{e^\epsilon-2}{e^\epsilon-1}\right) \boldsymbol{v} \right] \nonumber \\
    =&\Pr\left[ \mathcal{M}(t)=\left(1+\frac{1}{\psi}(1-\frac{e^\epsilon-2}{e^\epsilon-1})\right) \left(-\beta\psi \boldsymbol{v} \right) \right] \nonumber \\
    =&\Pr\left[ \mathcal{M}(t)=\left(1+\frac{1}{\psi}(1-\frac{e^\epsilon-2}{e^\epsilon-1})\right) \nabla_{\theta} \mathcal{L}_{\rm DPO} (t) \right] \nonumber \\
    =&\Pr\left[ \tilde{w}=1+\frac{1}{\psi}(1-\frac{e^\epsilon-2}{e^\epsilon-1}) \right] = \frac{e^\epsilon}{e^\epsilon+1}. \label{eq:g=r2}
\end{align}

For $t'$, we have
\begin{align}
    & \Pr\left[ \mathcal{M}(t')=-\beta\left(\psi+1-\frac{e^\epsilon-2}{e^\epsilon-1}\right) \boldsymbol{v} \right] \nonumber\\
    =&\Pr\left[ \mathcal{M}(t')=-\beta\left(1-\psi-1-\frac{1}{e^\epsilon-1}\right) (-\boldsymbol{v}) \right] \nonumber\\
    =&\Pr\left[ \mathcal{M}(t')=\left(1-\frac{1}{\psi'}(1+\frac{1}{e^\epsilon-1})\right) (-\beta)(-\boldsymbol{v})\psi' \right] \nonumber\\
    =&\Pr\left[ \mathcal{M}(t')=\left(1-\frac{1}{\psi'}(1+\frac{1}{e^\epsilon-1})\right) \nabla_{\theta} \mathcal{L}_{\rm DPO} (t') \right] \nonumber\\
    =&\Pr\left[ \tilde{w}'=1-\frac{1}{1-\psi} (1+\frac{1}{e^\epsilon-1}) \right] = \frac{1}{e^\epsilon+1}. \label{eq:g'=r2}
\end{align}

\noindent\textbf{Case 3: $\mathcal{M}(\cdot)\in r_3$.} It is straightforward to verify that the outputs of $\mathcal{M}(t)$ and $\mathcal{M}(t')$ never fall into $r_3$.  
Thus, for any $r\in r_3$, we have $\Pr[\mathcal{M}(t)=r]=\Pr[\mathcal{M}(t')=r]=0$.

\vspace{1mm}
\noindent\textbf{Putting it all together.} 
Combining (\ref{eq:g=r1}), (\ref{eq:g'=r1}), (\ref{eq:g=r2}), and (\ref{eq:g'=r2}) yields 
\begin{align*}
    e^{-\epsilon}  \leq \frac{\Pr\left[\mathcal{M}(t)=r\right]}{\Pr\left[\mathcal{M}(t')=r\right]} \leq e^\epsilon
\end{align*}
for any $r\in \mathcal{R}$, which establishes the $\epsilon$-preference privacy guarantee of $\mathcal{M}$. 
Finally, by the transformation invariance property of preference privacy (Proposition~\ref{corollary:transformation-inv}), Algorithm \ref{alg:privdpo} also satisfies $\epsilon$-preference privacy. This completes the proof.
\end{proof}

\subsection{Proof of Theorem \ref{thm:privdpo-grad-bound}}\label{appendix:proof-err-bound}
We first present a variant of Bernstein's inequality. The proof of Theorem \ref{thm:privdpo-grad-bound} established on this lemma.
\begin{lemma}[Bernstein's Inequality]\label{lemma:simplified-bern}
    Let $Z_1,\dots,Z_m$ be independent random variables with $\mathbb{E}[Z_i]=0$, $\sum_{i=1}^{m}\mathrm{Var}(Z_i)\leq V$, and $|Z_i|\leq M$. For any $\gamma\in(0,1)$, we have
    \begin{align*}
        \Pr\left[ \left\vert \sum_{i=1}^{m}Z_i \right\vert \leq \sqrt{2V\log\frac{2}{\gamma}}+\frac{M}{3}\log\frac{2}{\gamma} \right] \geq 1-\gamma.
    \end{align*}
\end{lemma}
\begin{proof}
    The standard Bernstein's inequality states that for any $t>0$, it holds that
    \begin{align}
        \Pr\left[ \left\vert \sum_{i=1}^{m}Z_i \right\vert \geq t \right]\leq 2\exp\left( -\frac{\frac{1}{2}t^2}{\sum_{i=1}^{m}\mathbb{E}[Z_i^2]+\frac{1}{3}Mt} \right).\label{eq:st-bern}
    \end{align}
    To enforce the right hand side of (\ref{eq:st-bern}) less or equal than $\gamma$, it suffices to let $t$ satisfy
    \begin{align*}
        \frac{\frac{1}{2}t^2}{V+\frac{1}{3}Mt}\geq \log\frac{2}{\gamma}.
    \end{align*}
    Note that $t=\sqrt{2V\log\frac{2}{\gamma}}+\frac{M}{3}\log\frac{2}{\gamma}$ is a valid choice of $t$ that satisfies the above inequality. Substituting this choice of $t$ into (\ref{eq:st-bern}) establishes the lemma.
\end{proof}

We now proceed to prove Theorem~\ref{thm:privdpo-grad-bound}.%
\begin{proof}[Proof of Theorem \ref{thm:privdpo-grad-bound}]
For the $i$th DPO triplet $(x,y_w,y_l)$,
let 
\[\boldsymbol{v}_i:= \nabla_\theta\log\pi_\theta(y_w\mid x) - \nabla_\theta\log\pi_\theta(y_l\mid x),\]
and 
\[\psi_i := \sigma \left(\beta\log\frac{\pi_{\theta}(y_l\mid x)}{\pi_{\rm ref}(y_l \mid x)} - \beta\log\frac{\pi_{\theta}(y_w\mid x)}{\pi_{\rm ref}(y_w \mid x)}\right).\]
Then the privatized gradient of the $i$th DPO triplet can be written as
$\boldsymbol{g}_i=-\beta \psi_i \tilde{w}_i\boldsymbol{v}_i $. 
Let the aggregate \privdpo{} gradients be $\boldsymbol{g}=\frac{1}{m}\sum_{i=1}^{m} \boldsymbol{g}_i$, and denote by $\boldsymbol{g}^*=-\frac{\beta}{m}\sum_{i=1}^{m}\psi_i \boldsymbol{v}_i$ the standard, non-private DPO gradient. We will prove the tail bound for the following error in any direction $\boldsymbol{u}\in \mathbb{S}^{d-1}$:
\[ {\rm Err}_{\boldsymbol{u}} := \boldsymbol{u}^{\top} (\boldsymbol{g} - \boldsymbol{g}^*).\]

Note that $\boldsymbol{g}_i^*=\mathbb{E}\left[ \boldsymbol{g}_i \mid \psi_i,\boldsymbol{v}_i \right]$, and it can be verified that the random weight $\tilde{w}_i$ in line \ref{algline:random-weight} of Algorithm \ref{alg:privdpo} satisfies $\mathbb{E}\left[ \tilde{w}_i \mid \psi_i \right]=1$ for any $i$.
Let $Z_i\in\mathbb{R}$ be the random variable defined as
\begin{align}
    Z_i :={} & \boldsymbol{u}^\top \left( \boldsymbol{g}_i - \boldsymbol{g}^*_i \right) \nonumber\\
    ={} & \boldsymbol{u}^{\top}\left(\boldsymbol{g}_i - \mathbb{E}\left[ \boldsymbol{g}_i \mid \psi_i,\boldsymbol{v}_i \right] \right) \nonumber \\
    ={} & -\beta \psi_i\left( \tilde{w}_i - \mathbb{E}[\tilde{w}_i \mid \psi_i] \right) \boldsymbol{u}^\top \boldsymbol{v}_i.\label{eq:zt}
\end{align}
Then the error of $\boldsymbol{g}$ can be written as the average of random variables $Z_i$ as follows
\[ {\rm Err}_{\boldsymbol{u}} = \frac{1}{m}\sum_{i=1}^{m} Z_i. \]
By construction, $\{Z_i\}$ are independent mean-zero random variables (conditional on $\{\psi_i,\boldsymbol{v}_i\}$)
, which implies that
\begin{align}
    \mathbb{E}\left[ \mathrm{Err}_{\boldsymbol{u}} \right]=\frac{1}{m}\sum_{i=1}^{m}\mathbb{E}\left[Z_i\right]=0.\label{eq:mean-zero-err}
\end{align}

Next, we derive the conditional variance of $Z_i$. First, by the perturbation rule of $\tilde{w}_i$, we have
\begin{align*}
    \mathrm{Var}(\tilde{w} \mid \psi_i) & = \frac{1}{e^\epsilon+1} \cdot \frac{e^\epsilon} {e^\epsilon+1} \cdot \left( \frac{e^\epsilon+1}{e^\epsilon-1} \cdot \frac{1}{\psi_i} \right)^2 \\
    & = \frac{e^\epsilon}{(e^\epsilon-1)^2}\cdot\frac{1}{\psi_i^2}.
\end{align*}
Therefore, the variance of $Z_i$ is given by
\begin{align*}
    \mathrm{Var}(Z_i \mid \psi, \boldsymbol{v}_i)&=\beta^2 \psi^2 \mathrm{Var}(\tilde{w} \mid \psi_i) \left( \boldsymbol{u}^\top \boldsymbol{v}_i \right)^2 \\
    &= \beta^2\frac{e^\epsilon}{(e^\epsilon-1)^2} \left( \boldsymbol{u}^\top \boldsymbol{v}_i \right)^2.
\end{align*}
Accordingly, the variance of 
$Z_t/m$ can be derived as
\begin{align}
 \sum_{i=1}^{m}\mathrm{Var}\left(\frac{Z_i}{m} \middle| \psi_i,\boldsymbol{v}_i \right) \nonumber
    &= 
    \frac{1}{m^2}  \sum_{i=1}^{m}\mathrm{Var}(Z_i \mid \psi_i, \boldsymbol{v}_i) \nonumber\\
    &= \frac{\beta^2 e^\epsilon}{m^2 (e^\epsilon-1)^2} \sum_{i=1}^{m}\left( \boldsymbol{u}^\top \boldsymbol{v}_i \right)^2.\label{eq:var-zt}
\end{align}

Next, we proceed to bound $\vert Z_i\vert$. 
By the probability distribution of $\tilde{w}_i$, we have
\begin{align*}
    \left\vert \psi_i \tilde{w}_i - \psi_i \right\vert \leq \max \left\{ 1-\frac{e^\epsilon-2}{e^\epsilon-1}, 1+\frac{1}{e^\epsilon-1} \right\} = 1+\frac{1}{e^\epsilon-1}.
\end{align*}
Therefore, plug the above bound into (\ref{eq:zt}), we have 
\begin{align}
    \vert Z_i \vert = \beta \cdot \big\vert  \psi_i \tilde{w}_i - \psi_i  \big\vert \cdot\left\vert\boldsymbol{u}^\top \boldsymbol{v}_i \right\vert \leq \beta \left( 1+\frac{1}{e^\epsilon-1}\right)\left\vert \boldsymbol{u}^\top \boldsymbol{v}_i \right\vert.\label{eq:|zi|}
\end{align}
Applying Bernstein's inequality (Lemma \ref{lemma:simplified-bern}) to the random variables 
$\mathrm{Err}_{\boldsymbol{u}}=\sum_{i=1}^{m} Z_i/m$ with values expressed in (\ref{eq:var-zt}) and (\ref{eq:|zi|}), we have that conditioned on $\{\psi_i, \boldsymbol{v}_i\}$, with probability at least $1-\gamma$, it holds that
\begin{align}
     \left\vert \mathrm{Err}_{\boldsymbol{u}} \right\vert \leq &\frac{\beta}{m} \sqrt{\frac{2 e^\epsilon}{(e^\epsilon-1)^2} \left( \sum_{i=1}^{m} (\boldsymbol{u}^\top \boldsymbol{v}_i)^2 \right) \log\frac{2}{\gamma}} \nonumber\\
     &\quad+ \frac{\beta \left( 1+\frac{1}{e^\epsilon-1} \right) \max_i\left\{\vert \boldsymbol{u}^\top \boldsymbol{v}_i \vert\right\} }{3m} \log\frac{2}{\gamma}. \nonumber
\end{align}
This completes the proof.
\end{proof}

\section{Variants of \privdpo{}}
\label{appendix:variant-privdpo}

As shown in Algorithm~\ref{alg:privdpo}, \privdpo{} is a flexible framework that can incorporate various perturbation mechanisms to achieve preference privacy. 
There exist DP mechanisms that provide unbiased perturbation for numerical values, notably Duchi's mechanism \cite{duchi2018minimax} and the Piecewise mechanism \cite{wang2019collecting}, which formed the starting point of our design exploration. However, they assume that the inputs lie in a continuous range, implying a fundamentally different adversary model that requires additional output noise.
We develop \privdpo{} variants that replace its core weight perturbation with these mechanisms, yielding strong competitive baselines. However, applying them in the DPO setting is non‑trivial, as ensuring both privacy and unbiasedness requires formal analysis of DPO's privacy‑leakage surface in Section~\ref{sec:sensitivity-analysis}.

\vspace{1mm}
\noindent\textbf{Duchi's mechanism}~\cite{duchi2018minimax}.
Given the (sensitive) preference intensity $\psi$ and privacy budget $\epsilon$ as the input, Duchi's mechanism first computes the probability 
\[
p = \frac{e^\epsilon - 1}{2 e^\epsilon + 2} \cdot \psi + \frac{1}{2}.
\]
Then it outputs the value $\frac{e^\epsilon + 1}{e^\epsilon - 1}$ with probability $p$, and $-\frac{e^\epsilon + 1}{e^\epsilon - 1}$ with probability $1 - p$.

\vspace{1mm}
\noindent\textbf{Piecewise mechanism~}\cite{wang2019collecting}. Given the preference intensity $\psi$ and the privacy budget $\epsilon$, it defines:
\[
C = \frac{e^{\epsilon/2} + 1}{e^{\epsilon/2} - 1}, \quad 
\ell = \frac{C + 1}{2} \cdot \psi - \frac{C - 1}{2}, \quad 
r = \ell + C - 1.
\]
Then with probability $\frac{e^{\epsilon/2}}{e^{\epsilon/2} + 1}$, it outputs a value sampled uniformly at random from $[\ell, r]$. With the remaining probability, it samples output uniformly from $[-C, C] \setminus [\ell, r]$.

\vspace{1mm}
\noindent\textbf{Remark on privacy semantics.}
While these two mechanisms ensure unbiasedness, they implicitly assume that by altering the preference signal in the input DPO example, the preference intensity $\psi$ can change within $[-1,1]$ arbitrarily, rather than only taking values in $\{\psi,\psi-1\}$. This essentially leads to a somewhat stronger privacy guarantee than $\epsilon$-preference privacy, yet weaker than standard DP. In contrast, our \privdpo{} accurately satisfies $\epsilon$-preference privacy without introducing unnecessary perturbation, thus outperforms all baselines under the same level of privacy budget.


\section{Details of Experimental Evaluation}\label{appendix:detail-exp}

\subsection{DP-SGD Setup}\label{appendix:exp-dp-sgd}
We present detailed privacy accounting and results analysis. 
Since each training example is used exactly once in our setup (which is also the common practice in standard DPO~\cite{rafailov2023direct,meng2024simpo,ethayarajh2024model}), there is no composition or subsampling-based privacy amplification; we therefore adopt the analytical Gaussian framework~\cite{balle2018improving} for tight noise accounting instead of subsampled \Renyi{}-DP approach~\cite{abadi2016deep,zhu2019poission,mironov2019r}, resulting in a Gaussian noise standard deviation of $\sigma=0.37$.

\subsection{Additional Experimental Results}\label{appendix:tables}
Tables~\ref{tb:main-tldr} and~\ref{tb:main-shp} present the experimental results on reward-based metrics for the TL;DR summarization and UltraFeedback benchmarks. These results are consistent with the findings in Section~\ref{subsec:exp-main}, showing that \privdpo{} significantly and consistently outperforms all baselines across configurations. 

Tables~\ref{tb:winrate-tldr} and~\ref{tb:winrate-ultra} report win-rate comparisons on the TL;DR summarization and UltraFeedback benchmarks. The results align with the conclusions in Section~\ref{subsec:exp-winrate}. 
Table~\ref{tb:memadv-appendix} reports MemAdv for \privdpo{} and all privacy-preserving baselines under privacy budget $\epsilon=1$.

\begin{table}[!tbp]
\centering
\caption{Win rate of \textbf{\privdpo{}} on TL;DR summarization (\%)}
\vspace{-3mm}
\label{tb:winrate-tldr}
\begin{small}
\begin{tabular}{lccc}
\toprule
\textbf{\privdpo{}} vs. & \textbf{Pythia-2.8B} & \textbf{Qwen2.5-3B} & \textbf{Llama3.2-3B} \\ 
\midrule
PrivSFT       & 76.7   & 72.2  & 86.3     \\\midrule
RR            & 67.4   & 59.1  & 78.9     \\
Duchi         & 56.9   & 55.6  & 54.7     \\
Piecewise     & 63.1   & 57.7  & 57.4     \\\midrule
DPO           & 48.8   & 43.2  & 47.2     \\
\bottomrule
\end{tabular}
\vspace{-1mm}
\end{small}
\end{table}

\begin{table}[!tbp]
\centering
\caption{Win rate of \textbf{\privdpo{}} on UltraFeedback (\%)}
\vspace{-3mm}
\label{tb:winrate-ultra}
\begin{small}
\begin{tabular}{lccc}
\toprule
\textbf{\privdpo{}} vs. & \textbf{Pythia-2.8B} & \textbf{Qwen2.5-3B} & \textbf{Llama3.2-3B} \\ 
\midrule                         \\
PrivSFT    &   61.3   &  72.2   &  71.5    \\\midrule
RR         &   55.5   &  66.0   &  66.7    \\
Duchi      &   54.8   &  55.7   &  56.5    \\
Piecewise  &   57.5   &  56.3   &  56.4    \\\midrule
DPO        &   49.1   &  49.3   &  49.6    \\
\bottomrule
\end{tabular}
\vspace{-1mm}
\end{small}
\end{table}

\begin{table}[!t]
\centering
\caption{Empirical memorization advantage ($\epsilon=1$)}
\vspace{-3mm}
\label{tb:memadv-appendix}
\begin{small}
\begin{tabular}{lccc}
\toprule
{Methods} & Anthropic-{HH} & {TL;DR}  & {UltraFeedback} \\ 
\midrule
\privdpo{}       &   0.58   &  0.62  &   1.02   \\
\midrule
RR     &   0.98   &  0.68  &   1.21   \\
Duchi      &   0.85   &  0.54  &   0.62   \\
Piecewise     &   0.56   &  0.61  &   0.53   \\
\midrule
DPO      &  9.76    &  7.36   &    9.58  \\
\bottomrule
\end{tabular}
\vspace{-2mm}
\end{small}
\end{table}

\begin{table*}[!t]
\centering
\caption{Performance comparison on TL;DR summarization dataset}
\vspace{-1mm}
\label{tb:main-tldr}
\begin{small}
\begin{tabular}{llccccccccc}
\toprule
\multirow{2}{*}{\textbf{Privacy Budget}} & \multirow{2}{*}{\textbf{Method}} & \multicolumn{3}{c}{\textbf{Pythia-2.8B}}                 & \multicolumn{3}{c}{\textbf{Qwen2.5-3B-Instruct}} & \multicolumn{3}{c}{\textbf{Llama3.2-3B-Instruct}} \\ 
\cmidrule(lr){3-5} \cmidrule(lr){6-8} \cmidrule(lr){9-11}

                         &                          & \textbf{RM} ($\uparrow$)  & \textbf{RA} ($\uparrow$)  & $\mathcal{L}_{\rm DPO}$ ($\downarrow$)   & \textbf{RM} ($\uparrow$)   & \textbf{RA} ($\uparrow$)  & $\mathcal{L}_{\rm DPO}$ ($\downarrow$)  & \textbf{RM} ($\uparrow$)  & \textbf{RA} ($\uparrow$)   & $\mathcal{L}_{\rm DPO}$ ($\downarrow$)   \\ \midrule
\multirow{4}{*}{$\epsilon=0.5$}  & RR          & 0.049  & 59.2\%  & 0.68  &  0.068  &  63.8\%   &  0.66  &  0.094  &  64.9\%  &  0.65 \\
                         & Duchi               & 0.091  & 58.8\%  & 0.67  &  0.215  &  64.9\%   &  0.63  &  0.261  &  65.6\%  &  0.62 \\
                         & Piecewise           & 0.034  & 55.8\%  & 0.69  &  0.175  &  63.7\%   &  0.64  &  0.238  &  65.2\%  &  0.62 \\
                         & \textbf{\privdpo{}} & \textbf{0.174} & \textbf{61.6\%} & \textbf{0.66} & \textbf{0.303} & \textbf{65.7\%} &  \textbf{0.61}  &  \textbf{0.394}  &  \textbf{67.7\%}  & \textbf{0.60} \\ \midrule
\multirow{4}{*}{$\epsilon=1$}   & RR           & 0.107  & 61.3\%  & 0.66  &  0.154  &  66.2\%   &  0.64  &  0.194  &  66.8\%  &  0.62 \\
                         & Duchi               & 0.178  & 61.4\%  & 0.66  &  0.302  &  65.3\%   &  0.63  &  0.394  &  67.1\%  &  0.61 \\
                         & Piecewise           & 0.141  & 61.1\%  & 0.67  &  0.268  &  65.1\%   &  0.63  &  0.363  &  66.7\%  &  0.61 \\
                         & \textbf{\privdpo{}} & \textbf{0.236} & \textbf{62.3\%} & \textbf{0.65} & \textbf{0.396} & \textbf{67.1\%} & \textbf{0.60} & \textbf{0.456} &  \textbf{68.9\%}   &  \textbf{0.59}  \\ \midrule
Non Private              & DPO                 & 0.291  & 62.5\%  & 0.65  &  0.497 & 67.5\% &  0.60 &  0.516 & 70.1\% &  0.57 \\ \bottomrule
\end{tabular}
\vspace{3mm}
\end{small}
\end{table*}

\begin{table*}[!tb]
\centering
\caption{Performance comparison on UltraFeedback dataset}
\vspace{-1mm}
\label{tb:main-shp}
\begin{small}
\begin{tabular}{llccccccccc}
\toprule
\multirow{2}{*}{\textbf{Privacy Budget}} & \multirow{2}{*}{\textbf{Method}} & \multicolumn{3}{c}{\textbf{Pythia-2.8B}}                 & \multicolumn{3}{c}{\textbf{Qwen2.5-3B-Instruct}} & \multicolumn{3}{c}{\textbf{Llama3.2-3B-Instruct}} \\ 
\cmidrule(lr){3-5} \cmidrule(lr){6-8} \cmidrule(lr){9-11}

                         &                          & \textbf{RM} ($\uparrow$)  & \textbf{RA} ($\uparrow$)  & $\mathcal{L}_{\rm DPO}$ ($\downarrow$)   & \textbf{RM} ($\uparrow$)   & \textbf{RA} ($\uparrow$)  & $\mathcal{L}_{\rm DPO}$ ($\downarrow$)  & \textbf{RM} ($\uparrow$)  & \textbf{RA} ($\uparrow$)   & $\mathcal{L}_{\rm DPO}$ ($\downarrow$)   \\ \midrule
\multirow{4}{*}{$\epsilon=0.5$}  & RR          & 0.061  & 58.7\%  & 0.67  &  0.123  &  67.4\%  &  0.64  &  0.111  &  65.9\%  &  0.65 \\
                         & Duchi               & 0.198  & 60.9\%  & 0.66  &  0.286  &  67.5\%  &  0.61  &  0.326  &  65.6\%  &  0.63 \\
                         & Piecewise           & 0.152  & 59.1\%  & 0.67  &  0.265  &  66.7\%  &  0.62  &  0.242  &  65.1\%  &  0.63 \\
                         & \textbf{\privdpo{}} & \textbf{0.282} & \textbf{62.8\%} & \textbf{0.63} & \textbf{0.502} & \textbf{68.2\%} &  \textbf{0.58}  &  \textbf{0.486}  &  \textbf{67.3\%}  & \textbf{0.60} \\ \midrule
\multirow{4}{*}{$\epsilon=1$}   & RR           & 0.153  & 64.5\%  & 0.64  &  0.218  &  70.4\%  & 0.61 &   0.219  &  68.2\%  &  0.61  \\
                         & Duchi               & 0.291  & 63.1\%  & 0.64  &  0.395  &  68.8\%  & 0.59 &   0.401  &  66.7\%  &  0.61  \\
                         & Piecewise           & 0.276  & 63.9\%  & 0.63  &  0.398  &  68.9\%  & 0.59 &   0.362  &  67.6\%  &  0.61  \\
                         & \textbf{\privdpo{}} & \textbf{0.357} & \textbf{65.8\%} & \textbf{0.61} & \textbf{0.538} & \textbf{70.8\%} & \textbf{0.57} & \textbf{0.578} &  \textbf{68.8\%}   &  \textbf{0.58}  \\ \midrule
Non Private              & DPO                 & 0.488  & 66.0\%  & 0.60  &  0.684 & 71.8\% &  0.55 &  0.703 & 71.5\% &  0.55 \\ \bottomrule
\end{tabular}
\end{small}
\end{table*}

\subsection{Attack}\label{appendix:attack}
This subsection describes the empirical attack used to compute the memorization advantage reported in Section~\ref{subsec:mem-gap}. 

\vspace{1mm}
\noindent\textbf{Attacker's knowledge.} 
The attacker is given a prompt $x$ and two candidate responses $(y_1,y_2)$ and aims to predict which response is preferred in the underlying dataset. This attacker corresponds to the natural inference task considered in label-DP-style settings, where features are public and only the label is sensitive.

\vspace{1mm}
\noindent\textbf{Attack intuition.}
The intuition behind the attack is that alignment training increases the relative likelihood of preferred responses compared to non-preferred ones.
If a model has learned and potentially memorized preference labels from training examples, this information may be reflected in the model's scoring of the two responses, making the preferred response easier to identify. 
By comparing how strongly an aligned model distinguishes between the two candidates relative to a prior model that has not seen preference labels, the attack probes whether training introduces additional, example-specific preference information beyond what is already inferable from public text or generalization.

\vspace{1mm}
\noindent\textbf{The attack pipeline.}
For an aligned model $\pi_{\theta}$ trained with DPO or \privdpo{}, we use the model's implicit reward w.r.t. the pre-trained base $\pi_\mathrm{base}$ model to score each candidate response, denoted by
\begin{align*}
    r_{\theta} = \log\pi_{\theta}(y \mid x) - \log\pi_\mathrm{base}(y \mid x).
\end{align*}
The attacker predicts $y_1 \succ y_2$ if $r_{\theta}(x,y_1)>r_{\theta}(x,y_2)$, and vice versa.

For the prior model $\pi_{\mathrm{prior}}$ (which has no associated reference or base model), we use the raw log-likelihood as the attack score
\begin{align*}
    s_\mathrm{prior} = \log \pi_\mathrm{base}(y\mid x),
\end{align*}
and predict preference by comparing $s_{\mathrm{prior}}(x,y_1)$ and $s_{\mathrm{prior}}(x,y_2)$. 

\vspace{1mm}
\noindent\textbf{Attack metric and interpretation.} 
Attack performance is measured using accuracy, i.e., the fraction of examples for which the attacker correctly predicts the dataset preference. 
Note that this attack is not intended to represent a worst-case adversary. Rather, it serves as a concrete lower-bound sanity check that complements the formal $\epsilon$-preference-privacy guarantee. In particular, high raw attack accuracy does not necessarily imply privacy leakage, as preferences may be inferable from public information alone. The MemAdv metric reported in Section~\ref{subsec:mem-gap} isolates the incremental advantage introduced by training and is therefore more informative for measuring memorization of private preference labels. 

\end{document}